\documentclass[reprint,amsfonts, amssymb, amsmath,  showkeys,prx, superscriptaddress, twocolumn,longbibliography,nofootinbib]{revtex4-2}
\usepackage{float}
\makeatletter
\let\newfloat\newfloat@ltx
\makeatother
\usepackage[english]{babel}
\usepackage[utf8]{inputenc}
\usepackage{graphics}
\usepackage{selinput}
\usepackage[normalem]{ulem}
\usepackage[shortlabels]{enumitem}

\usepackage{braket}
\usepackage{mathtools}
\usepackage{physics}
\usepackage{xcolor}
\usepackage{graphicx}
\usepackage[left=16mm,right=16mm,top=35mm,columnsep=15pt]{geometry} 
\usepackage{adjustbox}
\usepackage{placeins}
\usepackage[T1]{fontenc}
\usepackage{lipsum}
\usepackage{csquotes}
\usepackage{bm}

\usepackage[linesnumbered,ruled,vlined]{algorithm2e}

\usepackage{notes2bib}

\SetKwInput{kwInit}{Init}

\def\HC{\mathcal{H}}

\def\ad{^{\dagger}}

\newcommand{\fsnull}[1]{}
\newcommand{\old}[1]{}

\usepackage[makeroom]{cancel}
\usepackage[toc,page]{appendix}
\usepackage[colorlinks=true,citecolor=blue,linkcolor=magenta]{hyperref}

\usepackage{tikz}
\tikzset{every picture/.style=remember picture}

\usepackage[utf8]{inputenc}
\usepackage{graphicx}
\usepackage{xcolor}
\usepackage{amsmath}
\usepackage{amsthm}
\usepackage{bm}
\usepackage{bbm}
\usepackage{comment}
\usepackage{appendix}
\usepackage{mathdots}
\usepackage{lipsum}
\usepackage{verbatim}
\usepackage{natbib}
\usepackage{nccmath}
\usepackage{amsfonts} 

\newcommand{\inprod}[1]{\left\langle #1 \right\rangle}

\newcommand{\poly}{\operatorname{poly}}
\newcommand{\CSIM}{{\rm CSIM}}
\newcommand{\NCSIM}{\overline{{\rm CSIM}}}
\newcommand{\QESIM}{{\rm QESIM}}

\newcommand{\QSIM}{{\rm QSIM}}

\newcommand{\NBP}{\overline{{\rm BP}}}
\newcommand{\NQESIM}{\overline{{{\rm CSIM}_{\mathrm{QE}}}}}

\newcommand{\BC}{\mathcal{B}}
\newcommand{\CC}{\mathcal{C}}

\newcommand{\GC}{\mathcal{G}}

\newcommand{\OC}{\mathcal{O}}
\newcommand{\PC}{\mathcal{P}}

\newcommand{\SC}{\mathcal{S}}

\newcommand{\Var}{{\rm Var}}

\renewcommand{\geq}{\geqslant}
\renewcommand{\leq}{\leqslant}

\DeclareMathOperator*{\argmin}{arg\,min}
\renewcommand{\vec}[1]{\boldsymbol{#1}}  

\newcommand*{\id}{\openone}

\newcommand{\bs}{\textsf{BS}}

\newcommand{\liea}{\mathfrak{g}}

\newcommand{\thv}{\vec{\theta}}

\def\be{\begin{equation}}
\def\ee{\end{equation}}
\def\bs{\begin{split}}
\def\e{\end{split}}
\def\ba{\begin{eqnarray}}
\def\bea{\begin{eqnarray}}

\def\tea{\end{eqnarray}}
\def\ea{\end{eqnarray}}
\def\eea{\end{eqnarray}}

\newtheoremstyle{italicheader}
  {6pt}   
  {6pt}   
  {\normalfont}
  {}      
  {\itshape}
  {.}
  {0.5em}
  {}      

\theoremstyle{italicheader}

\newtheorem{theorem}{Theorem}

\newtheorem{claim}{Claim}

\usepackage{amssymb}
\usepackage{dsfont}

\def\be{\begin{equation}}
\def\te{\end{equation}}
\def\ee{\end{equation}}
\def\ba{\begin{eqnarray}}
\def\bea{\begin{eqnarray}}

\def\tea{\end{eqnarray}}
\def\ea{\end{eqnarray}}
\def\eea{\end{eqnarray}}

\begin{document}

\title{ Does provable absence of barren plateaus imply classical simulability?}

\author{M. Cerezo}\email{cerezo@lanl.gov}
\affiliation{Information Sciences, Los Alamos National Laboratory, Los Alamos, NM 87545, USA}
\affiliation{Quantum Science Center, Oak Ridge, TN 37931, USA}

\author{Martin Larocca}
\affiliation{Theoretical Division, Los Alamos National Laboratory, Los Alamos, NM 87545, USA}
\affiliation{Center for Nonlinear Studies, Los Alamos National Laboratory, Los Alamos, New Mexico 87545, USA}
 
\author{Diego Garc\'ia-Mart\'in}
\affiliation{Information Sciences, Los Alamos National Laboratory, Los Alamos, NM 87545, USA}

\author{N. L. Diaz}
\affiliation{Information Sciences, Los Alamos National Laboratory, Los Alamos, NM 87545, USA}
\affiliation{Departamento de F\'isica-IFLP/CONICET,
		Universidad Nacional de La Plata, C.C. 67, La Plata 1900, Argentina}

\author{Paolo Braccia}
\affiliation{Theoretical Division, Los Alamos National Laboratory, Los Alamos, NM 87545, USA}

\author{Enrico Fontana}
\affiliation{University of Strathclyde, Glasgow G1 1XQ, UK}

\author{Manuel S. Rudolph}
\affiliation{Institute of Physics, Ecole Polytechnique Fédérale de Lausanne (EPFL),  Lausanne CH-1015, Switzerland}

\author{Pablo Bermejo}
\affiliation{Donostia International Physics Center, Paseo Manuel de Lardizabal 4, San Sebasti\'an E-20018, Spain}
\affiliation{Information Sciences, Los Alamos National Laboratory, Los Alamos, NM 87545, USA}

\author{Aroosa Ijaz}
\affiliation{Theoretical Division, Los Alamos National Laboratory, Los Alamos, NM 87545, USA}
\affiliation{Department of Physics and Astronomy, University of Waterloo, Ontario, N2L 3G1, Canada}
\affiliation{Vector Institute, MaRS Centre, Toronto, Ontario, M5G 1M1, Canada}

\author{Supanut Thanasilp}
\affiliation{Institute of Physics, Ecole Polytechnique Fédérale de Lausanne (EPFL),  Lausanne CH-1015, Switzerland}
\affiliation{Chula Intelligent and Complex Systems, Department of Physics, Faculty of Science, Chulalongkorn University, Bangkok, Thailand, 10330}

\author{Eric R. Anschuetz}
\affiliation{Institute for Quantum Information and Matter, Caltech, Pasadena 91125, USA}
\affiliation{Walter Burke Institute for Theoretical Physics, Caltech, Pasadena 91125, USA}

\author{Zo\"e Holmes}
\affiliation{Institute of Physics, Ecole Polytechnique Fédérale de Lausanne (EPFL),  Lausanne CH-1015, Switzerland}

\begin{abstract}

A large amount of effort has recently been put into understanding the barren plateau phenomenon.  In this perspective article, we face the increasingly loud elephant in the room and ask a question that has been hinted at by many but not explicitly addressed: \textit{Can the structure that allows one to avoid barren plateaus also be leveraged to efficiently simulate the loss classically?}  We collect evidence—on a case-by-case basis—that many commonly used models whose loss landscapes avoid barren plateaus can also admit classical simulation, provided that one can collect some classical data from quantum devices during an initial data acquisition phase. This follows from the observation that barren plateaus result from  a curse of dimensionality, and that current approaches for solving them end up encoding the problem into some small, classically simulable, subspaces.
Thus, while stressing that quantum computers can be essential for collecting data, our analysis sheds doubt on the information processing capabilities of many parametrized quantum circuits with provably barren plateau-free landscapes.  We end by discussing the (many) caveats in our arguments including the limitations of average case arguments, the role of smart initializations, models that fall outside our assumptions, the potential for provably superpolynomial advantages and the possibility that, once larger devices become available, parametrized quantum circuits could heuristically outperform our analytic expectations.

\end{abstract}

\maketitle

\section{Introduction}

In recent years, the initial excitement attracted by variational quantum algorithms~\cite{cerezo2020variationalreview,bharti2021noisy,endo2021hybrid} and quantum machine learning~\cite{wiebe2014quantumdeep,schuld2015introduction,biamonte2017quantum,cerezo2022challenges,  di2023quantum, abbas2023quantum} has been tempered by the barren plateau phenomenon~\cite{mcclean2018barren,larocca2024review,marrero2020entanglement,sharma2020trainability,patti2020entanglement,pesah2020absence,uvarov2020barren,cerezo2020impact,uvarov2020variational,wang2020noise,abbas2020power,arrasmith2021equivalence,larocca2021diagnosing,holmes2021connecting,
cerezo2020cost,khatri2019quantum,zhao2021analyzing,liu2021presence,miao2023isometric,letcher2023tight,basheer2022alternating,suzuki2023effect,rudolph2023trainability,kieferova2021quantum,thanaslip2021subtleties,lee2021towards,shaydulin2021importance,holmes2021barren,leadbeater2021f,zhang2022escaping,martin2022barren,grimsley2022adapt,leone2022practical,sack2022avoiding,kashif2023impact,friedrich2023quantum,garcia2023deep,kulshrestha2022beinit,volkoff2021efficient,kashif2023unified,monbroussou2023trainability,cherrat2023quantum,fontana2023theadjoint,ragone2023unified,diaz2023showcasing,park2023hamiltonian,sannia2023engineered,thanasilp2022exponential,west2023provably,mao2023barren}. Namely, there is a growing awareness that a large class of quantum learning architectures exhibit loss function landscapes that concentrate exponentially in system size towards their mean value. On such landscapes, exponential resources are required for training, prohibiting the successful scaling of variational quantum algorithms. Hence, identifying architectures and training strategies that provably do not lead to barren plateaus has become a highly active area of research. Examples of such strategies include shallow circuits with local measurements~\cite{cerezo2020cost,pesah2020absence,khatri2019quantum,zhao2021analyzing,liu2021presence,miao2023isometric}, dynamics with small Lie algebras~\cite{larocca2021diagnosing,monbroussou2023trainability,cherrat2023quantum,fontana2023theadjoint,ragone2023unified,diaz2023showcasing,west2023provably,heredge2024prospects}, identity initializations~\cite{zhang2022escaping,park2023hamiltonian, wang2023trainability}, embedding symmetries into the circuit's architecture~\cite{larocca2022group,meyer2022exploiting,skolik2022equivariant,ragone2022representation,nguyen2022atheory,schatzki2022theoretical,zheng2021speeding,east2023all}, adding non-unital noise or intermediate measurements~\cite{sannia2023engineered,mele2024noise,fefferman2023effect, crognaletti2024estimates,deshpande2024dynamic}, and certain classes of quantum generative models~\cite{kieferova2021quantum,rudolph2023trainability,letcher2023tight}.

However, these strategies all, in some sense, make use of some simple underlying structure of the problem. This provokes the question: \textit{Could the very same  structure that allows one to provably avoid barren plateaus be leveraged to efficiently simulate the loss function classically?} Here we argue that the answer to this question is ``Yes and No''. Specifically, we argue and present strong evidence that a wide class of loss landscapes which provably do not exhibit barren plateaus can be simulated using either a classical algorithm or what we call a ``quantum-enhanced'' classical algorithm that runs in polynomial time. In the latter case, this simulation still necessitates the use of a quantum computer during an initial data acquisition phase~\cite{huang2021power,elben2022randomized,huang2021information,gyurik2023exponential}, but it does not require hybrid quantum-classical optimization loops. These arguments can be understood as  a soft form of dequantization of the information processing capabilities of many variational quantum circuits in barren plateau-free landscapes.  

Core to our argument is the observation that any loss based on evolving an initial state $\rho$ through a parameterized quantum circuit $U(\vec{\theta} )$ and then estimating the expectation value of an non-trivial observable $O$ can be written as an inner product between the Heisenberg evolved observable $U (\vec{\theta} )^\dagger O U (\vec{\theta} )$ and the state $\rho$. Given that both of these objects live in the exponentially large vector space of operators, one can generally expect this overlap to be --on average over $\thv$-- exponentially small. This is the essence of the barren plateaus phenomenon---the curse of dimensionality. If, however, the evolved observable is confined to a polynomially large subspace, then the loss becomes the inner product between two objects in this reduced space and can therefore avoid barren plateaus. But, in this case one can also simulate the loss by representing the initial state, circuit, and measurement operator as polynomially large objects contained in, and acting on, the small subspace.

Our general argument is supported by an analysis of widely used schemes through which we show that all considered methods for avoiding barren plateaus can be efficiently classically simulated, as well as by recent works which explicitly perform the classical simulations~\cite{bermejo2024quantum,angrisani2024classically,lerch2024efficient,anschuetz2024arbitrary,mele2024noise,shin2024dequantising,ermakov2024unified,miller2025simulation}.  Fundamentally, it is the very proof of absence of barren plateaus that allows us to identify the polynomially-sized subspaces in which the relevant part of the computation lives. Using this information, we can then determine the set of expectation values one needs to estimate  (either classically or quantumly) to enable classical simulations. In the latter case the quantum computer is still essential, but is accessed non-adaptively to generate a \textit{classical surrogate} rather than via a hybrid optimization loop.

Given the potential for misunderstanding, let us first state a few caveats to our claims. 
Firstly, our argument applies to widely used models and algorithms that employ a loss function formulated as the expectation of an observable for a state evolved under a parametrized quantum circuit, as well as variants using measurements of this form followed by classical post-processing. This encompasses the majority of popular quantum architectures including most standard variational quantum algorithms, many quantum machine learning models and certain families of quantum generative schemes. However, it does not cover all possible quantum learning protocols. 

Secondly, while for all our case studies it is possible to identify the ingredients necessary for simulation, we do not prove that this will always be possible. Thus, in principle, there could be models for which the landscape is free of barren plateaus, and yet we do not know how to simulate it. This could arise for sub-regions of a landscape which could be explored via smart initialization strategies, when the small subspace is otherwise unknown, or even when the problem lives in the full exponential space but is highly structured. Indeed, we provide an explicit (but highly contrived) construction for the latter and raise the possibility that more such cases might be found heuristically once larger quantum devices become available.

Finally, having identified these caveats, we present new opportunities and research directions that follow from our results. 
In particular, we discuss the potential offered by warm starts and the fact that even if polynomial-time classical simulation is available, the computational cost might still be too large, thus enabling potential polynomial advantages when running the variational quantum computing scheme on a quantum computer. More exotically, we suggest that by exploiting the structure of conventional fault-tolerant quantum algorithms, it might yet be possible to construct highly structured variational architectures for which superpolynomial quantum advantages can be realized.

\section{Definitions for  barren plateaus and simulability}

Variational quantum computing algorithms encode a problem of interest into an optimization task. The standard approach is to train a parametrized quantum circuit to minimize a loss function that quantifies the quality of the solution~\cite{cerezo2020variationalreview,bharti2021noisy,endo2021hybrid,wiebe2014quantumdeep,schuld2015introduction,biamonte2017quantum,cerezo2022challenges}. These algorithms are hybrid computational models in the sense that they use quantum hardware to obtain an estimate of the loss and then leverage the power of classical optimizers to determine parameter updates for the next set of experiments. 

In what follows, we will assume that the loss function takes the form 
\begin{equation}\label{eq:loss}
    \ell_{\thv}(\rho,O)=\Tr[U(\thv)\rho U\ad(\thv)O]\,.
\end{equation}
Here,  $\rho\in\BC$ is an $n$-qubit input state belonging to the set of bounded operators $\BC$ acting on a $2^n$-dimensional Hilbert space $\HC$, $O\in\BC$ is some non-trivial Hermitian operator (with $\norm{O}_\infty \leq 1$), $U(\thv)$ is the parametrized quantum circuit, and $\thv$ is a set of trainable parameters. Losses of this form can be used to tackle a wide range of problems through different choices of $\rho$, $O$ and $U(\thv)$. We note that while algorithms can employ more general loss functions that require computing multiple such quantities (e.g., by sending different states through the circuit, or by estimating the expectation value of several operators), we will focus on the fundamental case where the loss is given by Eq.~\eqref{eq:loss}, as the lessons derived here can be extrapolated to other scenarios.

\begin{figure*}[t]
    \centering
\includegraphics[width=1\linewidth]{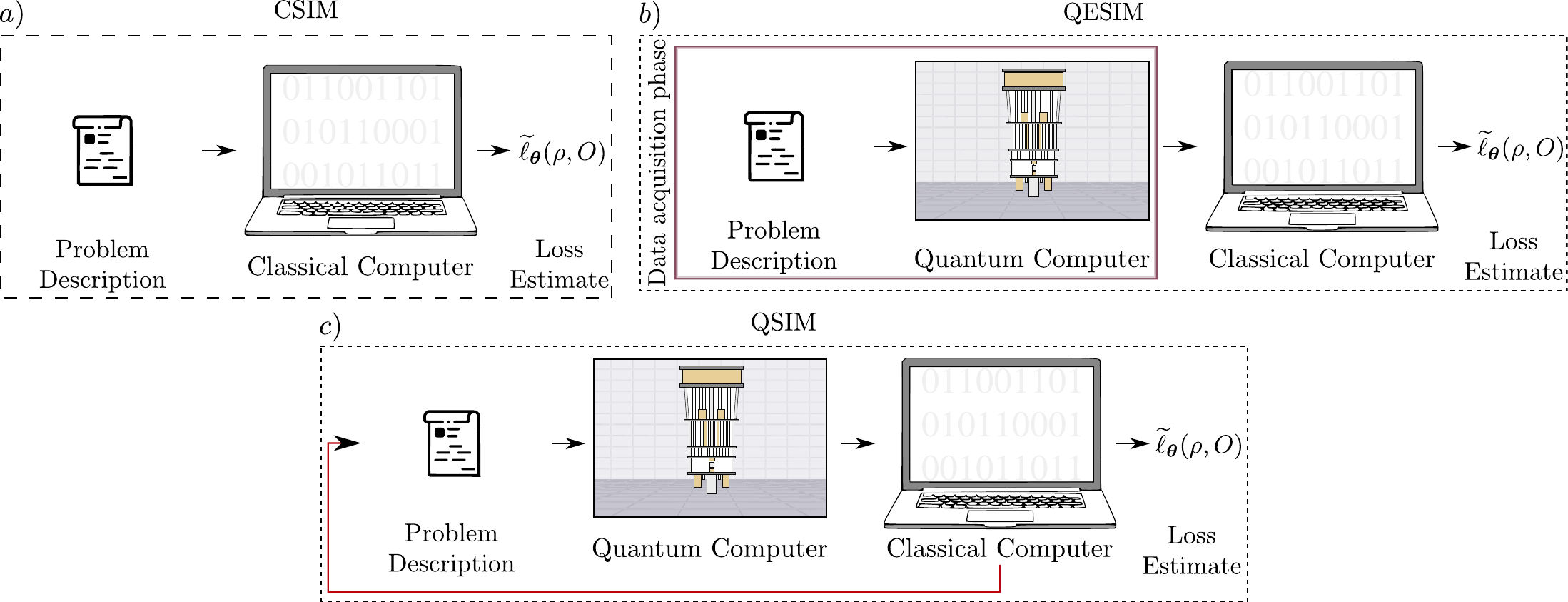}
    \caption{\textbf{Schematic description of the simulation classes.} a) Problems in $\CSIM$ are those for which there exists a fully classical algorithm that takes as input the problem description and estimates the loss function in polynomial time with a  classical computer. Here, access to a quantum computer is not needed. b) Problems in  $\QESIM$ 
 are those where one is allowed access to a quantum computer for an initial data acquisition phase that takes no more than polynomial time. At the end of this phase, access to the quantum computer is no longer allowed. A classical algorithm then takes the problem's description, and the data obtained from the quantum device, to estimate the loss function in polynomial time. Note that one can compute problems in $\QESIM$ without needing to run a parametrized quantum circuit on the quantum hardware. c) Problems in $\QSIM$ are those where one allows `on demand' access to a quantum computer. Here, one usually implements the parametrized quantum circuit on the device.}
    \label{fig:models}
\end{figure*}

To better understand and classify the problems we are focusing on, we find it convenient to define \textit{problems} specified by \textit{classes} $\mathcal{C} = \{\mathcal{I}\}$ of \textit{problem instances}.
A problem instance $\mathcal{I}$ is determined by an efficiently-sampleable parameter distribution $\mathcal{P}$, and some efficient \textit{classical description} of $\rho$, $U(\thv)$, and $O$ that can be used to estimate $\ell_{\thv}(\rho,O)$ on a quantum computer in polynomial time. These could be a quantum circuit that prepares $\rho$ from some fiducial state, a dictionary of the gate types and placements in $U(\thv)$, and the Pauli decomposition of $O$. We assume that these descriptions can be encoded in a string of size in $\OC(\poly(n))$.

For concreteness, let us consider an example of a problem class. Let  $\mathcal{C}_{{\rm shallowHEA}}$ be the class of all instances where the circuit is a one-dimensional hardware efficient ansatz~\cite{kandala2017hardware,cerezo2020cost} composed of $L\in\OC(\log(n))$ layers of two-qubit gates acting on neighboring qubits in a brick-layered structure. Moreover, we take  $\rho$ to be an $n$-qubit state preparable by a circuit with $\OC(\poly(n))$ gates when acting on the all-zero state, and $O$ some Pauli operator that is diagonal in the computational basis (e.g., $O=Z^{\otimes n}$ or $O=Z_{\mu}$ for some  $\mu \in \{1,\ldots,n\}$). Here, we further assume that all the gates in the circuit are parametrized, and every parameter is sampled uniformly at random.

Over the past few years, there has been a tremendous amount of work put forward to understand if the loss functions in a given  problem class based on variational circuits are \textit{trainable} (and we refer the reader to Ref.~\cite{gil2024relation} for a subtle discussion of what trainable, or even variational, even means). Several sources of untrainability have been detected~\cite{cerezo2022challenges} such as the presence of sub-optimal local minima~\cite{bittel2021training,fontana2022nontrivial,anschuetz2022beyond,anschuetz2021critical,larocca2021diagnosing} and expressivity limitations~\cite{tikku2022circuit}. However, the vast majority of trainability analysis has been concentrated around the \textit{barren plateau} phenomenon~\cite{larocca2024review,marrero2020entanglement,sharma2020trainability,patti2020entanglement,pesah2020absence,uvarov2020barren,cerezo2020impact,uvarov2020variational,wang2020noise,abbas2020power,arrasmith2021equivalence,larocca2021diagnosing,holmes2021connecting,
cerezo2020cost,khatri2019quantum,zhao2021analyzing,liu2021presence,miao2023isometric,rudolph2023trainability,letcher2023tight,basheer2022alternating,suzuki2023effect,kieferova2021quantum,thanaslip2021subtleties,lee2021towards,shaydulin2021importance,holmes2021barren,leadbeater2021f,zhang2022escaping,martin2022barren,grimsley2022adapt,leone2022practical,sack2022avoiding,kashif2023impact,friedrich2023quantum,garcia2023deep,kulshrestha2022beinit,volkoff2021efficient,kashif2023unified,monbroussou2023trainability,cherrat2023quantum,fontana2023theadjoint,ragone2023unified,diaz2023showcasing,park2023hamiltonian,sannia2023engineered,thanasilp2022exponential,west2023provably,mao2023barren} (see Ref.~\cite{larocca2024review} for a review on barren plateaus). 
When a problem exhibits a barren plateau, its loss function becomes---on average---exponentially concentrated with the system size~\cite{mcclean2018barren,cerezo2020cost}.
We say a class $\mathcal{C}$ of problems is  \textit{provably barren plateau-free}, i.e., $\mathcal{C} \in \NBP$, if one can show that 
\begin{equation}
    \Var_{\thv\sim\mathcal{P}}[ \ell_{\thv}(\rho,O) ] \in\Omega\left( \frac{1}{\mathrm{poly}(n)}\right),
\end{equation}
for all loss functions $\ell_{\thv}(\rho,O)$ in the class $\mathcal{C}$ of parameterized quantum circuits. 
Note that one can define a more general class of barren plateau-free problems where an explicit proof of absence of barren plateaus is not needed, but for now we will consider this restricted class.

Proving that certain types of problems are in $\NBP$ has recently become an active area of research. While such studies are extremely important, it is worth noting that just because the loss functions for a given problem are barren plateau-free does not mean that they are practically useful or that they can achieve a quantum advantage. For instance, one should wonder if the quantum computer is being employed in a meaningful way. That is, one would ideally like to show that beyond absence of barren plateaus, there exists no classical algorithm that can also efficiently compute the loss. 

To better tackle this question, we then first need to define what it means to \textit{compute} a loss function, and also what we understand by \textit{classical simulability}. In particular, since the notion of barren plateaus is an average statement over the landscape, in this context the natural notions for computing and simulating the loss will also be average ones. However, stronger notions of simulability, such as one where one guarantees that the loss can be computed for all points on the landscape, will also be discussed below.

First, let us define the task of \textit{computing a loss function}, such as that in Eq.~\eqref{eq:loss}, for a given problem. We will say that an algorithm can compute the loss functions in an instance $\mathcal{I}$ if,   with high probability $\thv\sim\PC$, it can implement a function $\tilde{\ell}_{\thv}(\rho,O)$ approximating the loss up to error  $\epsilon$, i.e., 
\begin{equation}
     \left( \tilde{\ell}_{\thv}(\rho,O) - \ell_{\thv}(\rho,O) \right)^2 \leq \epsilon\,.
\end{equation}
An algorithm can compute the loss functions in the problem~$\mathcal{C}$ if it can compute them for all instances $\mathcal{I}$ in $\mathcal{C}$. One could also define problems in terms of being able to compute the loss function $\ell_{\thv}(\rho,O)$ and its gradients $\nabla \ell_{\thv}(\rho,O)$.
Being able to access the derivative information can be useful during the parameter training process.
While in some cases computing the loss allows us to access derivatives (e.g., via parameter-shift rule~\cite{mitarai2018quantum,schuld2019evaluating}), we will restrict ourselves to only estimating the loss function. 

As previously mentioned, the above definition guarantees that we can compute the loss function with high probability when the parameters are sampled according to $\PC$. While simulating--with high probability--random point in the landscape is typically easy (see for instance~\cite{angrisani2024classically}) it may not be particularly useful as it does not necessarily allow us to train the circuit. One might also be interested in computing the loss given \textit{any} parameter settings. Thus, we will also consider a stronger version where an algorithm can compute the loss function in an  instance $\mathcal{I}$ if  \textit{for all}   $\thv$, it can implement a function $\tilde{\ell}_{\thv}(\rho,O)$ approximating the loss up to error  $\epsilon$. Of course, more pragmatic definitions exist where we do not care about approximating the loss function in average, or across all points, but where we instead focus on actually solving the optimization task at hand. As such, one may wish to obtain a $\tilde{\ell}_{\thv}(\rho,O)$ such that the solution to the optimization problem $\argmin_{\thv}\tilde{\ell}_{\thv}(\rho,O)$ is as good (according to some metric) as that obtained by solving $\argmin_{\thv}\ell_{\thv}(\rho,O)$. We will not focus on this case here but we note that its exploration is an open area of research~\cite{gil-fuster2024understanding,bermejo2024quantum}.

Now that we have defined what it means to compute (with high probability\ or with certainty over $\bm{\theta}$) a loss, we can present different notions of what it means to simulate it. 
We begin with the most basic and intuitive definition for classical simulability, which we simply dub \textit{Classical Simulation} (CSIM).
A problem $\mathcal{C}$ is in $\CSIM$ if a polynomial-time classical algorithm can compute every instance in $\mathcal{C}$.
As schematically shown in Fig.~\ref{fig:models}, CSIM is performed entirely on a classical device with no need nor access to a quantum computer.

We also consider classical simulation algorithms enhanced by polynomial-size data obtained from quantum experiments, which we denote as \emph{Quantum Enhanced Classical Simulation} ($\QESIM$), which morally means that the quantum landscape can be \textit{classically surrogated}. 
In this case, one is given a problem instance and is allowed to use a quantum computer for an initial data acquisition phase, which takes no more than polynomial time (and uses no more than polynomial-memory).
During this phase, one can prepare copies of the initial quantum state, apply some operations, and independently measure them via some efficient tomographic or classical shadow techniques \cite{huang2020predicting,elben2022randomized,anshu2023survey}. Such a procedure can be used to obtain an efficiently storable classical representation (i.e., storable in string of size in $\OC(\poly(n))$) of the state, unitary or measurement operator.
Once this initial phase is over, one cannot access the quantum device anymore, and the computation of $\ell_{\thv}(\rho,O)$  has to be done purely classically.
A problem $\mathcal{C}$ is in $\QESIM$ if a polynomial-time classical algorithm, which can utilize data obtained from quantum devices in an initial data acquisition phase, can compute every instance in $\mathcal{C}$. Here we note that the problem of estimating a loss with data from a quantum computer can be cast into a decision problem. In this case, $\CSIM$ can be related to BPP, $\QESIM$ is closely connected to  BPP/Samp~\cite{huang2021power} and even more closely connected to BPP/qgenpoly~\cite{jerbi2023shadows}, and $\QSIM$ is related to BQP.  Speaking more informally, algorithms in $\QESIM$ are also often described as classical surrogate methods.

We stress that all problems in $\CSIM$ and $\QESIM$ do not require running a parametrized quantum circuit on the quantum computer. Quantum resources are either entirely unnecessary ($\CSIM$), or used only in some initial data acquisition phase ($\QESIM$).

\begin{figure}[t]
    \centering
\includegraphics[width=.7\columnwidth]{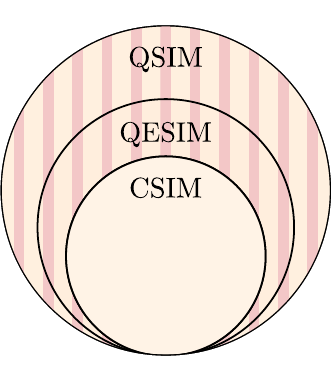}
    \caption{\textbf{Class inclusions.} We show the inclusions between the classes $\CSIM$, $\QESIM$ and $\QSIM$. The region outside of $\CSIM$ (highlighted with stripes) corresponds to problems where a quantum advantage could potentially be achieved.  }
    \label{fig:inclusions}
\end{figure}

This leads us to define \textit{Quantum Simulation} (QSIM).
A problem class $\mathcal{C}$ is in $\QSIM$ if a polynomial-time quantum algorithm can compute all instances in $\mathcal{C}$.  As depicted in Fig.~\ref{fig:models}, models in QSIM  allow for feedback between the classical and quantum computer. Moreover, the models in QSIM will usually require implementing the parametrized quantum circuit on quantum hardware.

This is a convenient point to make a few important remarks. First, as shown in Fig.~\ref{fig:inclusions}, we note that, by definition, the following inclusions hold:
\begin{equation}
    \CSIM\subseteq\QESIM\subseteq\QSIM\,.
\end{equation}
For instance, $\CSIM \subseteq \QESIM$ because if the loss is fully classically simulable then it can be estimated by simply skipping the quantum computer.
The rest of the inclusions follow similarly. Clearly, QSIM is \textit{not} the largest possible set, and any loss that requires exponential time to estimate, even with a quantum computer, is beyond QSIM. 

A quantum advantage is possible for problems where the loss can be simulated only if we have access to a quantum computer, i.e., for any problem in $\QSIM \cap \NCSIM\,$. In fact, any problem where the loss is in $\QESIM$ but not $\CSIM$ is already capable of a quantum advantage as it requires a quantum device. The problems that fall in $\QESIM \cap \NCSIM$ may well be the most suited for implementation on a near-term quantum computer as the data acquisition phase could be less noisy than fully implementing the parametrized quantum circuit~\cite{mcclean2017hybrid, parrish2019quantum, bharti2020quantum, huang2021power, jerbi2023shadows, gyurik2023limitations}.

With the definitions above, we are ready to ask the main question that motivates this work: \textit{Are all provably barren plateau-free loss functions   also classically simulable (given polynomial-size data)?}
That is, if:
\begin{equation}
    \NBP \overset{?}{\subset} \CSIM \,\, \mathrm{or} \,\, \NBP\overset{?}{\subset} \QESIM\,.
\end{equation}

\section{What leads to absence of barren plateaus?}

To understand whether barren plateau-free losses are simulable, one must first understand the conditions leading to non-exponential concentration. While the study of barren plateaus was initially limited to case-by-case analyses, recent results have transformed our understanding of this phenomenon~\cite{fontana2023theadjoint,ragone2023unified,diaz2023showcasing}. As such, what previously seemed like a fragmented patchwork of special cases has started to coalesce into a cohesive unified theory of barren plateaus. In what follows, we will attempt to give an intuition for the sources of barren plateaus, and concomitantly how to avoid them. 

To begin, let us note the simple, yet extremely important, fact that the loss function can be re-written as 
\begin{equation}\label{eq:inner-product}
    \ell_{\thv}(\rho,O)=\inprod{\rho(\thv),O}=\inprod{\rho,O(\thv)}\,,
\end{equation}
where we have defined $\rho(\thv)=U(\thv)\rho U\ad(\thv)$, $O(\thv)=U\ad(\thv)OU(\thv)$, and where  $\inprod{A,B} = \Tr [A^\dagger B]$ denotes the  Hilbert--Schmidt inner product~\cite{knapp2013lie}. At a first glance,  Eq.~\eqref{eq:inner-product} indicates that the loss is expressed as the inner product---a similarity measure---between two (exponentially large) operators of $\BC$. This fact should already raise some red flags as one can expect that, under quite general assumptions, the inner product between two exponentially large objects will be (on average) exponentially small and concentrated. For instance, we refer the reader to~\cite{ragone2023unified,larocca2024review} for a formalization of such argument. As such,  problems with loss functions such as those in Eq.~\eqref{eq:inner-product} can generically be expected to have barren plateaus.  

\begin{figure}[t]
    \centering
\includegraphics[width=1\linewidth]{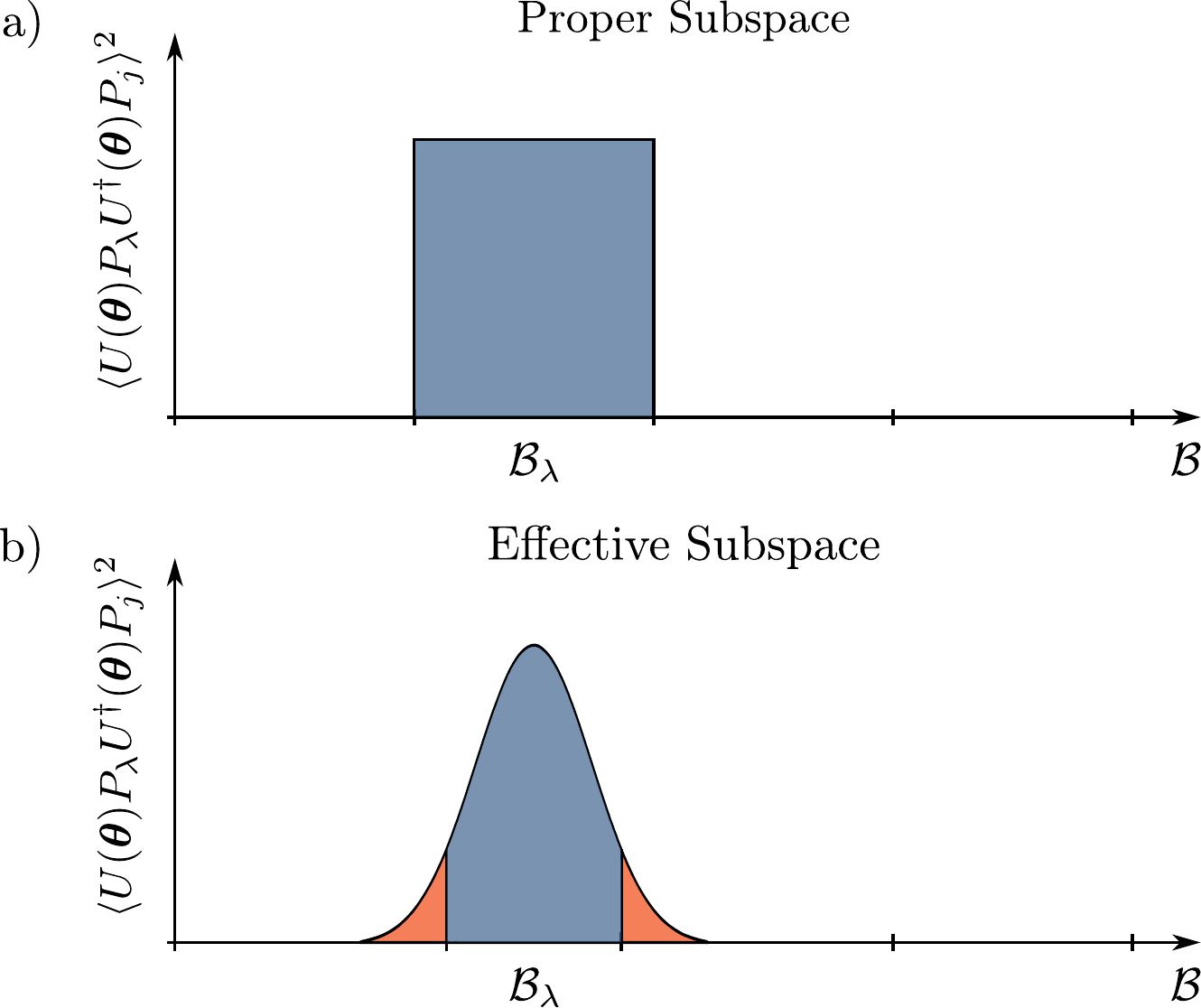}
\caption{\textbf{Adjoint action subspaces.} a) Given an operator $P_\lambda$ and some parameters $\thv$, we say $\BC_\lambda\subset\BC$ is a proper subspace if  $
\inprod{U(\thv)P_\lambda U\ad(\thv),P_j}^2$ is non-zero only for operators $P_j$ in $\BC_\lambda$.  b) We will instead say that $\BC_\lambda$ is an effective subspace if $\inprod{U(\thv)P_\lambda U\ad(\thv),P_j}^2$ is non-zero for many operators outside of $\BC_\lambda$, but is only large for operators in $\BC_\lambda$.
Proper and effective subspaces can arise either for all $\thv$, or with high probability when sampling $\thv$ from $\PC$.
}
    \label{fig:modules}
\end{figure}

If, however, the unitary $U(\thv)$ possesses additional structure, such as respecting some symmetry, then the loss can inherit this structure and potentially avoid the aforementioned issues. 
In particular, let us consider the adjoint action of $U(\thv)$ over the operator space $\BC$, and let us analyze if it leads to either subspaces, or, effective subspaces. That is, given some  $\thv$ and operator $P$ in an appropriate orthogonal basis of $\BC$ (such as a Pauli operator), we want to know: \textit{Where can $U(\thv)PU\ad(\thv)$ go?} Note that one can, and should, also ask the question: \textit{Where can $U\ad(\thv)PU(\thv)$ go?} But for simplicity of notation we will consider the first question.

Mathematically, this question can be answered by computing the inner products $
\inprod{U(\thv)PU\ad(\thv),P_j}$ for  the rest of the basis elements $P_j$. For instance, if the inner product is non-zero for all $P_j$, then we know that the unitary can transform $P$ into an operator that can spread out and reach operators across all of $\BC$. On the other hand, it could happen that the adjoint action of the unitary can only reach certain operators in $\BC$. For instance, as shown in Fig.~\ref{fig:modules}(a),  the adjoint action of $U(\thv)$ can lead to  well-defined subspaces in  $\BC$ such that $
\inprod{U(\thv)PU\ad(\thv),P_j}$ is only non-zero for operators in some $\BC_P \subset \BC$. This case can arise for instance in circuits with small Lie algebraic modules~\cite{larocca2021diagnosing,schatzki2022theoretical,kazi2024analyzing,anschuetz2022efficient,fontana2023theadjoint,ragone2023unified,diaz2023showcasing,heredge2024prospects} or in shallow-depth hardware efficient ans\"{a}tze~\cite{cerezo2020cost,khatri2019quantum,zhao2021analyzing,liu2021presence}.

A second case of interest arises when we have that 
$\inprod{U(\thv)PU\ad(\thv),P_j}$ is non-zero for a wide range of operators, but is large only for operators in a given subspace  $\BC_P$. This case is depicted schematically in Fig.~\ref{fig:modules}(b) and we will say that given $P$, the adjoint action of $U(\thv)$ leads to `effective' subspaces, rather than proper ones. 
We note that effective subspaces can either arise for all $\thv$ or with high probability for $\thv \sim \mathcal{P}$. In the latter case, for most $\theta$ values for the $\inprod{U(\thv)PU\ad(\thv),P_j}$ are large only within a small subspace, but for some low probability $\thv$ values, large overlaps outside this subspace could be observed. Effective subspaces appear, with high probability for  $\thv\sim\PC$, for quantum convolutional neural networks~\cite{cong2019quantum,pesah2020absence}, in circuits with small angle initialization~\cite{zhang2022escaping,park2023hamiltonian, wang2023trainability}, or even due to the effects of noise~\cite{mele2024noise}. 

To study whether or not a given problem lives within a subspace, it is convenient to express  $O$ in an orthogonal basis of $\BC$ as 
\begin{equation}\label{eq:O-decomp}
O=\sum_\lambda c_\lambda P_\lambda\,.   
\end{equation} 
Then,  denote as $\BC_\lambda\equiv\BC_{P_\lambda}$ the subspace associated to each $P_\lambda$ that is induced by the adjoint action of $U(\thv)$.  Given a state $\rho\in\BC$, we define its projection onto $\BC_\lambda$ as $\rho_\lambda=\sum_{P_j\in\BC_\lambda}a_j P_j$ with $a_j=\inprod{\rho,P_j}/\sqrt{\inprod{P_j,P_j}}$.
In this notation, the loss  function becomes:
\begin{align}
    \ell_{\thv}(\rho,O)&=\sum_\lambda c_\lambda \inprod{\rho,P_\lambda(\thv)}=\sum_\lambda c_\lambda \inprod{\rho_\lambda,P_\lambda(\thv)}\,,\label{eq:loss-subspaces}
\end{align}
which reveals that $\ell_{\thv}(\rho,O)$ is the sum of the inner products in each subspace. Note that in the previous equation we expanded the measurement operator and expressed the loss in terms of the subspaces obtained by Heisenberg evolving each basis element. However, one could also expand $\rho$ and study the loss in terms of the subspaces that concomitantly arise.

\begin{figure}[t]
    \centering
\includegraphics[width=.7\linewidth]{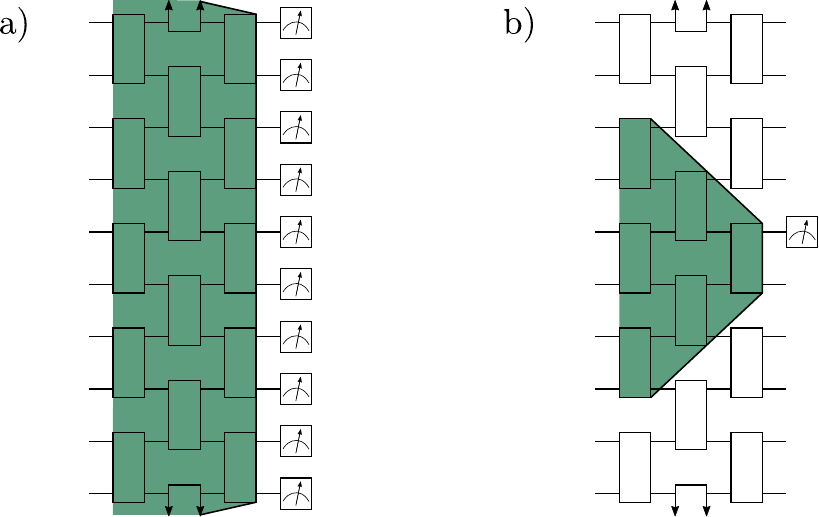}
\caption{\textbf{Subspaces for shallow hardware efficient ans\"{a}tze.} We consider classes of problems where the unitary $U(\thv)$ is an $L$-layered hardware efficient ansatz with two-qubit gates acting on alternating pairs of neighboring qubits in a brick-like fashion. We further assume that $L\in\OC(\log(n))$. a) For a global operator such as $O=Z^{\otimes n}$, the subspace obtained by adjoint action of $U(\thv)$ is exponentially large $\forall L$. b) Given a local operator, such as $O=Z_\mu$, one can see that  for all $\thv$ the ensuing subspace is proper and only contains Pauli operators acting on at most $\OC(\log(n))$ neighboring qubits. Hence, this subspace is only polynomially large. Colored regions depict the backwards light cone of the measurement operator when Heisenberg evolved~\cite{leone2022practical}.  }
    \label{fig:HEA}
\end{figure}

If any of the subspaces appearing in Eq.~\eqref{eq:loss-subspaces} is only polynomially large, then we can see that some component of the loss arises from  comparing  objects (via their Hilbert-Schmidt inner product)  in non-exponentially large spaces. In what follows,  we will use $\CC_{{\rm polySub}}$ to denote the class of problems where the action of  $U(\thv)$ on some term of the measurement operator is (either with high probability for  $\thv\sim\PC$ or for all $\thv$) contained in an identifiable polynomially scaling (proper or effective) subspace, i.e., such that $\text{dim}(\BC_\lambda)\in\OC( \text{poly}(n)) $, and where a basis for $\BC_\lambda$ can be classically obtained.

To make these ideas clearer, let us go back to the class of problems with shallow hardware efficient ans\"{a}tze $\mathcal{C}_{{\rm shallowHEA}}$. Consider first the case where the measurement operator is global, meaning that it acts on all qubits, such as $O=Z^{\otimes n}$. As schematically shown in Fig.~\ref{fig:HEA}(a), by applying the circuit to this operator one can obtain exponentially many other Pauli operators acting non-trivially on all qubits, meaning that the associated subspace is exponentially large. On the other hand, as seen in Fig.~\ref{fig:HEA}(b), if the measurement is a local operator such as $O=Z_\mu$, then  due to the bounded light cone structure of the circuit, it can only be mapped to Paulis acting on at most $2L\in\OC(\log(n))$ neighboring qubits. Since there are $\OC(\poly(n))$ such operators, the resulting subspace is a proper polynomial-sized subspace for any $\thv$. Hence, if $O$ contains any local term, we will have that $\mathcal{C}_{{\rm shallowHEA}}\subset \CC_{{\rm polySub}}$.

As mentioned above, showing that certain problem classes are in  $\CC_{{\rm polySub}}$ indicates that some part of the loss arises from comparing objects in polynomially large spaces. While this appears to be  a necessary step towards non-concentrated loss functions, it is also clearly \textit{not} a sufficient one.  In particular,  consider a  subspace $\BC_\lambda$ (for some $P_\lambda $ in Eq.~\eqref{eq:O-decomp}) such that $\dim(\BC_\lambda)\in\OC(\poly(n))$, but  $\rho$ or $O$ have almost no component when projected down into $\BC_\lambda$. For instance, if   $|\inprod{\rho_\lambda,\rho_\lambda}|$ or $|c_\lambda^2\inprod{P_\lambda,P_\lambda}|$  are in $\OC(1/\exp(n))$, then we will  clearly have  an issue since the signal in the loss coming from the polynomially-sized subspace can therefore still be exponentially weak. If the previous occurs for all small subspaces, then the loss will be obtained by comparing objects in exponentially large spaces, which we already know can lead to concentration issues.

\begin{table*}[t]
    \centering
\begin{tabular}{|c|c|c|c|}\hline
Problem instance $\CC$ based  on &Refs.&Operators in  the polynomial-sized $\BC_\lambda$&$O$ and $\rho$ leading to  $\CC\subset\NBP$ \\\hline\hline
Shallow hardware efficient ansatz   & \cite{cerezo2020cost,khatri2019quantum,zhao2021analyzing,liu2021presence,letcher2023tight,basheer2022alternating,suzuki2023effect} &  $\OC(\log(n))$ neighboring qubits (P) & Local $O$, area law  $\rho$ \\\hline
Generic shallow locally circuits    & \cite{napp2022quantifying, zhang2023absence}\footnote{See also the Supplemental Information for a detail proof of this claim.}  &  $\OC(1)$-weight qubits (E) & Local $O$, area law  $\rho$ \\\hline
Quantum convolutional neural network  & \cite{pesah2020absence} &  $\OC(1)$-weight  qubits (E) & Local $O$, area law $\rho$\\\hline
$U(1)$-equivariant   & \cite{larocca2021diagnosing,monbroussou2023trainability,cherrat2023quantum} &  Proj. with  $\OC(1)$ Hamming weight (P) & Equivariant $O$,  $\rho\in\BC_\lambda$ \\\hline
$S_n$-equivariant   & \cite{schatzki2022theoretical} & Permutation equivariant (P) & $O\in\BC_\lambda$, $\Tr[\rho_\lambda^2]\in\Omega(1/\poly(n))$\\\hline
Matchgate circuit   & \cite{diaz2023showcasing} & Product of $\OC(1)$ Majoranas (P) & $O\in\BC_\lambda$, $\Tr[\rho_\lambda^2]\in\Omega(1/\poly(n))$\\\hline
Small angle initialization    & \cite{zhang2022escaping,park2023hamiltonian, wang2023trainability}  & $\OC(\text{poly}(n))$ Pauli Operators (E)&Local $O$, area law $\rho$\\\hline
Small Lie algebra $\mathfrak{g}$ & \cite{fontana2023theadjoint,ragone2022representation,west2023provably} & Operators in $\mathfrak{g}$ (P) & $O\in i\mathfrak{g}$, $\Tr[\rho_{\mathfrak{g}}^2]\in\Omega(1/\poly(n))$\\\hline
Non-unital noisy-circuits  & \cite{mele2024noise} &  $\OC(\log(n))$   qubits (E) & Local $O$, any $\rho$ \\\hline
Dynamic circuits & \cite{deshpande2024dynamic} & $\OC(1)$-weight qubits (E) & Local $O$, any $\rho$\\\hline
Quantum generative modeling\footnote{Here we refer to quantum circuit Born machines with the maximum mean discrepancy loss~\cite{rudolph2023trainability} and quantum generative adversarial networks  with certain families of discriminators ~\cite{letcher2023tight} (both using shallow hardware efficient circuits).} & \cite{rudolph2023trainability, letcher2023tight} & Tensor networks (e.g., MPS) (P) & $O$, $\rho$ computational basis proj.\\\hline
\end{tabular}
    \caption{\textbf{Problem classes in $\NBP$ and also in $\CC_{{\rm polySub}}$.} In this table, we present the problems considered in our analysis. In all cases, we also provide the references where absence of barren plateaus was proven, as well as the relevant polynomially large subspace $\BC_\lambda$ (indicating whether it is proper (P) or effective (E)), and the states and measurement operators for which $\CC$ is in $\NBP$. We refer the reader to Table~\ref{tab:sim-algs} and the Supplemental Information for additional details on how the classical simulation is performed. Here `proj.' is used as a short-hand for projector, while MPS stands for matrix product states~\cite{verstraete2008matrix}. }
    \label{tab:train_models}
\end{table*}

To showcase the importance of having measurement operators and initial states which are ``well-aligned''  with the polynomial subspace, let us consider again the class of problems with shallow hardware efficient ans\"{a}tze. We already saw that if $O$ is local (or contains relevant local terms), then $\mathcal{C}_{{\rm shallowHEA}}\subset \CC_{{\rm polySub}}$. Thus, to know whether the loss will concentrate, we need to study  the Hilbert-Schmidt norm of the projected $\rho$ into the subsets of $2L\in\OC(\log(n))$ adjacent qubits in the backwards light cone of the local measurement. For instance, it is not hard to check that if $\rho$ is pure and follows a volume law of entanglement~\cite{cerezo2020cost,leone2022practical}, then $|\inprod{\rho_\lambda,\rho_\lambda}|$ will be exponentially close to zero, and the loss will be exponentially concentrated. On the other hand, if $\rho$ is pure and satisfies an area law of entanglement~\cite{eisert2010colloquium,leone2022practical}, then the previous inner product will be, at most, polynomially vanishing. Putting these realizations together we see that 
\begin{equation}
    \CC_{{\rm shallowHEA}}\subset \NBP\,,\,\, \text{if $\rho$ follows area law and $O$ is local}.\nonumber
\end{equation}
Note that these results actually correspond to a reinterpretation of Theorems 1 and 2 in Ref.~\cite{cerezo2020cost}, where it was shown that shallow-depth hardware efficient ans\"{a}tze lead to barren plateaus for global measurements, but are barren plateau-free if the measurement operator is local and the initial state follows an area law of entanglement.

The previous example of $\mathcal{C}_{{\rm shallowHEA}}$ illustrates an important point. We started with a problem class, identified the measurement operators which, when evolved under the adjoint of $U(\thv)$, remain in polynomially small subspaces, and then used this to determine the states for which the problem does not exhibit a barren plateau. This argument then leads to the question: \textit{How general is the connection between absence of barren plateaus and polynomially small subspaces arising from the unitary's adjoint actions?}  To address this, we have performed a detailed analysis of widely used barren plateau-free models and found that in all cases, it is \textit{precisely}  the existence of such subspaces which allows us to avoid exponential concentration. As such, we believe that the following claim is true for all widely used architectures and techniques:

\begin{claim}\label{claim1}
    \textit{[Standard provably barren plateau-free architectures live in classically identifiable polynomial subspaces.]} For all standard problem  classes $\mathcal{C}$ such that  $\mathcal{C}  \subset \NBP$, we have that
    \begin{equation}
         \mathcal{C} \subset \CC_{{\rm polySub}} \, .
    \end{equation} 
    That is, take a problem class which provably avoids barren plateaus. Then, if one studies the parametrized unitary's adjoint action on the measurement operation and initial state, one will find that it generates operators that live either exactly, or approximately, in a polynomially sized subspace in operator space. Moreover, such subspaces can be identified classically. 
\end{claim}

By `standard problems' we here refer to the conventional variational quantum architectures which have been proven to be in $\NBP$. 
In Table~\ref{tab:train_models} we present a non-exhaustive list of such architectures. Therein, we report the relevant polynomial subspace, indicating whether it is proper or effective, as well as the conditions on the measurement operator and initial state necessary for absence of barren plateaus. We emphasize that in many cases, it was the proof of absence of barren plateau itself which allowed us to determine the reported information. The majority of strategies studied, including shallow hardware efficient ansatz and highly symmetrized models, lead to proper subspaces for all $\thv$. The most important examples of models where one lives in effective subspaces (with high probability $\thv\sim\PC$) are quantum convolutional neural networks and small angle initialization strategies.

\section{Connection between absence of barren plateaus and simulability}

In the previous section we claimed that barren plateaus arise as a curse of dimensionality for  loss functions that compare objects in exponentially large spaces. We then argued that our attempts to fix this issue have ultimately led us to encode the problem in some polynomially small subspace, which we can classically identify. Here we show that the existence of such a subspace can be exploited to classically simulate the loss.

To illustrate our arguments, let us again start by considering a  shallow hardware efficient ansatz class where $O=Z_\mu$.  Since the measurement is local the adjoint action of $U(\thv)$ leads to a proper subspace $\forall \thv$ of  Pauli operators acting on at most $2L\in\OC(\log(n))$-neighboring qubits. From Fig.~\ref{fig:HEA}(b), we can then see that if we `drop' all the gates from $U(\thv)$ that are outside of the measurement's backwards light cone, the loss function remains unchanged.  By denoting as $U_\lambda(\thv)$ the reduction of  $U(\thv)$ that acts only on the  $2L$ qubits in the backwards light cone, we find that 
\begin{equation}
    \ell_{\thv}(\rho,Z_\mu)= \inprod{\rho,U\ad(\thv)Z_\mu U(\thv)}= \inprod{\rho_\lambda,U_\lambda\ad(\thv) Z_\mu U_\lambda(\thv)}\,\nonumber.
\end{equation}
In the last equality, one computes the inner product between objects living, and acting on, the polynomially-sized subspace. 

Indeed, we have taken an important step in the right direction as classically computing the loss now requires working with non-exponential operators. Still, we need to determine what $\rho_\lambda$ is.  If $\rho$ is some product state, then one can classically find its projection onto the subspace by computing (e.g., via pen and paper) the expectation value for all operators $P_j\in\BC_\lambda$. However, if $\rho$ is some state obtained at the output of a given circuit, then there is no generic strategy which allows us to obtain these expectation values classically (even if we are promised that $\rho$ is an area law state~\cite{anshu2023survey}). On the other hand, given access to a quantum computer, one can efficiently estimate $\rho_\lambda$  during the initial data acquisition phase. In particular, here one simply prepares the quantum state, and measures the expectation value for all (polynomially many) operators in a basis of $\BC_\lambda$. One could also simply perform standard classical shadow tomography~\cite{huang2020predicting,elben2022randomized}, and this information will suffice.

Hence, for local observables and area law initial states, we have that
\begin{equation}
   \CC_{{\rm shallowHEA}}\subset \CSIM\,, \quad \text{or}\quad  \CC_{{\rm shallowHEA}}\subset \QESIM\nonumber\,,
\end{equation}
and it is clear that the loss can be estimated without ever needing to use a quantum computer to run the parametrized quantum circuit. Hence,  no problem based on shallow hardware efficient ans\"{a}tze with a local measurement can be outside of $\QESIM$. We note that this result was originally reported in Ref.~\cite{basheer2022alternating}.It is worth additionally highlighting that, given the simplicity of the requisite measurement procedures here, one could lift the restriction that $\rho$ is preparable by a poly depth quantum circuit and also consider quantum states that result from quantum experiments running in polynomial time, e.g., from analog or non-unitary processes. This makes the quantum-enhanced classical simulation method more flexible than directly implementing the parameterized quantum circuit on hardware~\cite{jerbi2023power}.

Here we argue that the process described above for simulating shallow HEA circuits can be applied to any problem in $\CC_{{\rm polySub}}$. In all cases, the simulation follows the following three steps:
\begin{enumerate}
    \item Identify the subspaces $\BC_\lambda$ of polynomial dimension.
    \item Characterize the adjoint action of $U(\thv)$ on the basis element $P_\lambda$ in the decomposition of  $O$ (or  $\rho$) that lead to the $\BC_\lambda$. 
    \item Compute or measure the component of $\rho$ (or $O$) that is in the relevant polynomial  $\BC_\lambda$ subspaces.
\end{enumerate}

Identifying the polynomial subspace $\BC_\lambda$ is the most important step and usually requires understanding the internal properties of $U(\thv)$ and how these translate into the circuit's adjoint action. For instance, in the shallow hardware efficient ansatz example, we needed to realize that the relevant subspace is composed only of operators acting on the qubits in the backwards light cone. More generally, given a problem in $\NBP$, one can obtain $\BC_\lambda$ by carefully analyzing the proof techniques used to show non-exponential concentration and reverse engineering what the exploited structure is. In particular, all proofs for the absence of barren plateaus come with fine print, in the sense that they only hold for specific choices of $\rho$ and $O$. Thus, most of the work has already been performed in proving absence of barren plateaus and one can  take the initial states and measurement operators for which the proof holds and use these to infer $\BC_\lambda$. 

\begin{table*}[t]
    \centering
\begin{tabular}{|c|c|c|c|}\hline
Problem instance $\CC$ based  on & Tomographic procedure for $\rho$ & Simulation algorithm based on & Simulable \\\hline\hline
Shallow hardware efficient ansatz &  Pauli classical shadows~\cite{huang2020predicting}    & Light-cone sim. reduced $U(\thv)$ & $\forall \, \thv$ \\\hline
Generic shallow locally circuits   & Pauli classical shadows    & Pauli Propagation~\cite{angrisani2024classically} &  $\thv\sim\PC$  \\\hline
Quantum convolutional neural network\footnote{In the Supplemental Information we also show that certain parameter restricted quantum convolutional neural networks can be simulated for all $\thv$. For a general simulation of vanilla quantum convolutional neural networks  see~\cite{bermejo2024quantum}.} & Pauli classical shadows &  Pauli Propagation & $\thv\sim\PC$   \\\hline
$U(1)$-equivariant  & Computational basis measurement& Givens Rotations~\cite{kerenidis2021classical}   & $\forall \,  \thv$  \\\hline
$S_n$-equivariant     & Permutation invariant shadows~\cite{toth2010permutationally}& $\mathfrak{g}$-sim~\cite{goh2023lie}&$\forall \,  \thv$\\\hline
Matchgate circuit    & Expectation value of Pauli operators & $\mathfrak{g}$-sim &   $\forall \,  \thv$ \\\hline
Small angle initialization     & Pauli measurements & Tensor Networks~\cite{shin2024dequantising}, Pauli Prop.~\cite{lerch2024efficient} &  $\thv\sim\PC$  \\\hline
Small Lie algebra $\mathfrak{g}$  & Expectation value of algebra elements & $\mathfrak{g}$-sim & $\forall \,  \thv$ \\\hline
Non-unital noisy-circuits  & Not needed & Light-cone sim.~\cite{mele2024noise}, Pauli Prop.~\cite{schuster2024polynomial}  & $\thv\sim\PC$ \\\hline
Dynamic circuits & Pauli classical shadows & Pauli Propagation~\cite{angrisani2024classically} & $\thv\sim\PC$ \\\hline
Quantum generative modeling\footnote{Here we refer to quantum circuit Born machines with the maximum mean discrepancy loss~\cite{rudolph2023trainability} and quantum generative adversarial networks with certain families of discriminators ~\cite{letcher2023tight} (both using shallow hardware efficient circuits).}  & Not needed & Tensor Networks & $\forall \,  \thv$\\\hline
\end{tabular}
    \caption{\textbf{Simulation algorithm for the problems in $\NBP$.} In this table we present the problem instances considered in our analysis and an example algorithm that we can use to simulate them. Alternatives to tensor networks can also be used, e.g. Refs.~\cite{nemkov2023fourier,fontana2023classical, rudolph2023classical, shao2023simulating,beguvsic2023simulating}. We also provide the tomographic procedure that could be used for a generic initial state as well as indicating whether the simulation is available on average, or for all points in the landscape.  }
    \label{tab:sim-algs}
\end{table*}

Once we have identified $\BC_\lambda$, we can proceed to study the adjoint action of $U(\thv)$ over the relevant $P_\lambda$ operator. The key insight here is  to note that $U(\thv)P_\lambda U\ad(\thv)$ can only   map to operators in $\BC_\lambda$, and hence, can be somehow represented by its effective action on this subspace. The specific construction of this effective action is case dependent, but as detailed in the Supplemental Information (which contains additional details and proofs, which contains Refs.~\cite{wiersema2020exploring,jozsa2008matchgates,wan2022matchgate,de2013power,diaz2023parallel,bravyi2004lagrangian,dias2023classical,cudby2023gaussian,gigena2015entanglement,arrazola2022universal,johri2021nearest,lopez2022symmetric,somma2006efficient,zeier2011symmetry,wiersema2023classification,kokcu2024classification,aguilar2024full,vznidarivc2022solvable,deneris2024exact,braccia2024computing,belkin2023approximate,mittal2023local,benedetti2019generative,coyle2020born,alcazar2020classical,benedetti2019adversarial,perdomo2018opportunities,zoufal2021generative,ferris2012perfect,stoudenmire2010minimally,markov2008simulating,verstraete2006matrix,evenbly2011tensor,beguvsic2023fast,wecker2015progress,caro2021generalization,liu2023model,hur2021quantum,umeano2023can,mele2023introduction}.), we have been able to derive it for all considered barren plateau-free problems. In some cases, such as the aforementioned hardware efficient ansatz, one can trivially reduce $U(\thv)$ to a smaller dimensional unitary acting on $2L$ qubits by simply discarding all of the gates that are not in the light cone of the measurement  (see Fig.~\ref{fig:HEA}(b)). In some other cases, one needs to employ classical simulation algorithms based on tensor networks, operator truncation, or Lie algebraic techniques~\cite{anschuetz2022efficient,fontana2023classical,rudolph2023classical,goh2023lie}. We refer the reader to the Supplemental Information for additional details. It is worth noting that constructing the adjoint action of $U(\thv)$ on $\BC_\lambda$ has a computational cost associated with it, which one needs to take into consideration to estimate the simulation cost. Importantly, while we have found algorithms whose cost is polynomial in the number of qubits, it can nevertheless have poor scaling.

The final task is to obtain the component of $\rho$ and $O$ in $\BC_\lambda$, which we respectively denote as  $\rho_\lambda$ and $O_\lambda$. While it is possible to find the polynomially large subspaces and the effective adjoint of $U(\thv)$ fully classically, this will not be generically true  for determining  $O_\lambda$ and $\rho_\lambda$ (as we saw, for example, for the shallow HEA). 
While $O_\lambda$ is not much of a concern---since the classical description of $O$ could already contain this information---let us focus on obtaining $\rho_\lambda$. Recalling that, given an orthogonal basis $\{P_j\}_j$, one can write $\rho_\lambda=\sum_{P_j\in\BC_\lambda}\inprod{\rho,P_j}P_j/\sqrt{\inprod{P_j,P_j}}$,  then we see that one needs to compute the expectation values $\inprod{P_j,\rho}$. While for simple states such as the all-zero state this task can be performed fully classically, for general input states $\rho$ this might not be the case.

With the previous arguments, we state the following claim.

\begin{claim}\label{claim2}
     \textit{[Problems in known polynomial subspaces are classically simulable (potentially requiring data from a quantum computer).]} 
    \noindent Consider a problem class such that $\CC \subset \CC_{{\rm polySub}}$. 
     Then either  
    \begin{equation}
    (i.) \, \, \, \, \CC \subset \CSIM\,\quad \nonumber
    \end{equation}
    if $O_\lambda$, $\rho_\lambda$ can be obtained classically, or
    \begin{equation}
    (ii.) \, \, \, \, \CC \subset \QESIM\,\quad \nonumber
    \end{equation}
    if $O_\lambda$, $\rho_\lambda$ need to  be estimated
   on a quantum computer. \\

   \noindent In particular, this opens the door towards the simulation of barren plateau-free loss functions $\CC \subset\NBP$ living in a known polynomial subspace (as described in Claim~\ref{claim1}) using a classical algorithm without the need to implement a parametrized quantum circuit on a quantum computer.
\end{claim}

Again, for all barren plateau-free problems considered here, we have determined  the relevant subspace, found efficient procedures to estimate the projections of the initial state and measurement operator, as well as for the action of the unitary. A summary of the  resulting simulation algorithm and measurement protocols are reported in Table~\ref{tab:sim-algs} and expanded upon in the Supplemental Information.

As mentioned above, and as indicated in Table~\ref{tab:sim-algs}, for problems that live in proper polynomial subspaces (which constitute the vast majority of barren plateau-free schemes considered in the literature), one can simulate the loss in the strong sense, meaning that there exists a  classical algorithm that approximates the loss for any $\thv$. In the case of problems that live in effective subspaces, the classical simulability is available with high probability over $\thv \sim \mathcal{P}$. In particular,  here we can simulate the parts of the loss that do not have barren plateaus. In the next  section we will discuss the caveats resulting from this slightly weaker claim, but also opportunities to mitigate its disadvantages.

\medskip 
To wrap up, we highlight that combining Claim~\ref{claim1} and Claim~\ref{claim2} with our case-by-case analysis we can  arrive at the conclusion: \textit{For the families of problems and architectures with provable absence of barren plateaus we analyzed, it does not seem that the parametrized quantum circuit need to be implemented on a quantum computer in order to estimate the loss in polynomial time}. 
Note that our results do \textit{not} imply a dequantization of variational quantum computing as a whole, since a quantum device might still be needed for the initial data acquisition phase. Viewed this way, our results can be viewed positively as highlighting the potential of a different learning paradigm where the quantum computer is used non-adaptively to create a classical surrogate of the loss landscape. Still, in a way, Claims~\ref{claim1} and~\ref{claim2} \textit{do}, in a sense, dequantize the variational part of the model and shed some serious doubts on the non-classicality of the information processing being done by a barren plateau-free loss function.

 \section{Caveats and future directions}\label{sec:futureop}

In this section, we present several caveats to our arguments, as well as interesting new research directions.

\subsection{Caveats}

First and foremost, we would like to highlight the fact that our general arguments are based on intuition gathered from a case-by-case study of widely used circuit architectures and techniques (although we do refer the reader to~\cite{bermejo2024quantum,angrisani2024classically,lerch2024efficient,anschuetz2024arbitrary,mele2024noise,shin2024dequantising,ermakov2024unified, schuster2024polynomial} for explicit simulation algorithms). It was not possible to analyze every work claiming absence of barren plateaus---we studied only those in Table~\ref{tab:train_models}---but we highly encourage the community to check if our results are as widely applicable as we believe them to be. Moreover, while we have argued that absence of barren plateaus is linked to the presence of relevant polynomially small subspaces, we do not close the door to the existence of other non-exponential concentration mechanisms. For instance, Claim~\ref{claim1} states \textit{for all cases studied} that, if $\mathcal{C}  \subset \NBP$ then $\mathcal{C} \subset \CC_{{\rm polySub}}$. We do \textit{not} claim that  $\NBP \subset  \CC_{{\rm polySub}}$. Thus we are \textit{not} claiming that \textit{any} non-concentrated loss can always be classically simulated.

In fact, one can construct examples of non-concentrated loss functions that are not classically simulable (see the Supplemental Information for one such example based on cryptographic hardness). Crucially, these examples do not resemble current mainstream variational quantum algorithms but instead draw inspiration from conventional fault-tolerant quantum algorithms for which we expect a superpolynomial quantum advantage to be achievable. In doing so, they break our initial assumption that comparing objects living in exponentially large spaces generically leads to concentrated expectation values, as the circuits used therein are not generic but rather purposely constructed. In fact, textbook quantum algorithms do not suffer from the curse of dimensionality as the exponentially large quantum states are manipulated in a well-thought and orderly manner, rather than a variational one. Nonetheless, as evidenced by our rather contrived examples in the Supplemental Information, it is not obvious how to use these techniques in variational quantum computing as textbook quantum algorithms are only built for very specific tasks. 

Another caveat is that the amount of quantum resources required by problems that fall into $\QESIM$ varies on a case-by-case basis. Many cases, as shown in Table~\ref{tab:sim-algs}, require only Pauli measurements and thus the load on the quantum device is very light. Indeed, a universal quantum computer may not be required and rather a more bespoke analogue simulator or quantum experiment may suffice. At the other extreme one could consider a deep quantum circuit with only a single trainable parameter $\phi$~\cite{gil-fuster2024understanding}. Such a circuit could be constructed so that it does not have a barren plateau and yet also fall into $\QESIM$, as one can trivially construct a classical surrogate for the landscape by running the circuit in `a data collection phase' at different $\phi$ values. The absence of barren plateau results for small \textit{patches} of quantum landscapes close to minima~\cite{haug2021optimal, puig2024variational, mhiri2025unifying} which can also be surrogated by running the parametrized quantum circuit and making appropriate measurements~\cite{lerch2024efficient}. These cases thus strictly align with our claims that problems that can provably avoid barren plateaus fall into $\QESIM$ but stretch the original spirit of $\QESIM$ by requiring substantially more quantum resources from the quantum device. 

It is worth noting that our general arguments against conventional variational quantum algorithms require provable absence of barren plateaus, as the proof itself is used in our derivation of the classical simulability techniques. As such, we cannot comment on situations where one can heuristically find large gradients (e.g., via numerics~\cite{grimsley2022adapt,dborin2022matrix,rudolph2022synergy}) but where the relevant subspace is not known. Such situations could arise for very special parts of the landscape, such as warm-starts or other smart initialization strategies. Although we believe that a careful analysis of many such cases will highlight that the problem is essentially contained in a polynomially-sized subspace. For instance, we can refer to the specific case of ADAPT-VQE~\cite{grimsley2022adapt}, which has no formal proof of absence of barren plateaus but has been heuristically shown to have large gradients for close to identity initializations. Here, it has been recently shown that Majorana Propagation~\cite{miller2025simulation} could efficiently simulate the action of the ADAPT-VQE circuit. Of course, these are preliminary results and there is as of yet no efficient protocol for characterizing the adjoint action of $U(\thv)$ throughout training and guaranteeing that it will remain in the simulable region. Indeed, as suggested by our counterexamples, it might be that the clever initializations explore an exponentially large space in some structured---but unknown---manner.

Relatedly, let us highlight that, for problems that live in effective but not proper subspaces, classical simulation is possible for $\thv\sim\mathcal{P}$ only with high probability. While this can enable for the simulation of randomly chosen points in certain regions of the landscape~\cite{lerch2024efficient,mhiri2025unifying}, this need not be necessarily useful for training, as one may be in a non-interesting part of the landscape, or the training might take us towards regions that not classically simulable. 
More concretely, if the subspace is effective (as in  Fig.~\ref{fig:modules}(b)),  our methods will provide simulation techniques that will only be faithful in the vicinity of the initialization. However, if the relevant sector $\BC_\lambda$ shifts during training the simulation might cease to be reliable. The extent to which this is possible is yet to be explored, but recent results have shown that this phenomenon does not seem to occur in quantum convolutional neural networks~\cite{bermejo2024quantum} when benchmarked in standard tests, as there one initializes --and remains during training-- in the same classically simulable region of the landscape.

Finally, we would like to again highlight that our analysis was oriented towards loss functions based on quantities such as that in Eq.~\eqref{eq:loss}. Hence, our claims are not directly applicable to problems such as quantum Boltzmann machines~\cite{coopmans2023sample} where one employs different types of loss functions. In fact, quantum Boltzmann machines are another example of a variational model built upon a quantum primitive (thermal state preparation) which one expects is hard to simulate classically. Hence this is an interesting avenue for balancing quantum advantage and trainability considerations.

\subsection{New opportunities}

Perhaps the main goal of this perspective article is to provide new and exciting research directions that take into account the connection between  barren plateaus and simulability. As such, we here discuss several new opportunities that we think could be fruitful to pursue.

First and foremost, we would like to note that just because some loss function can be classically simulated in polynomial time, this does not mean that it is still practical to do so. In particular, some of the simulation algorithms we have found (see the Supplemental Information) could still be prohibitively expensive to implement. Thus, by embracing the fact that some barren plateau-free loss functions are classically simulable, one could compare the computational complexity of the simulation versus that of estimating the loss function in a quantum computer, and thus potentially find provable polynomial speed-ups or at least more favorable constant factors with the same polynomial scaling. We refer the reader to Ref.~\cite{anschuetz2022interpretable, miao2023convergence, anschuetz2024arbitrary} for some examples on this research direction.

\begin{figure*}[t]
    \centering
\includegraphics[width=1\linewidth]{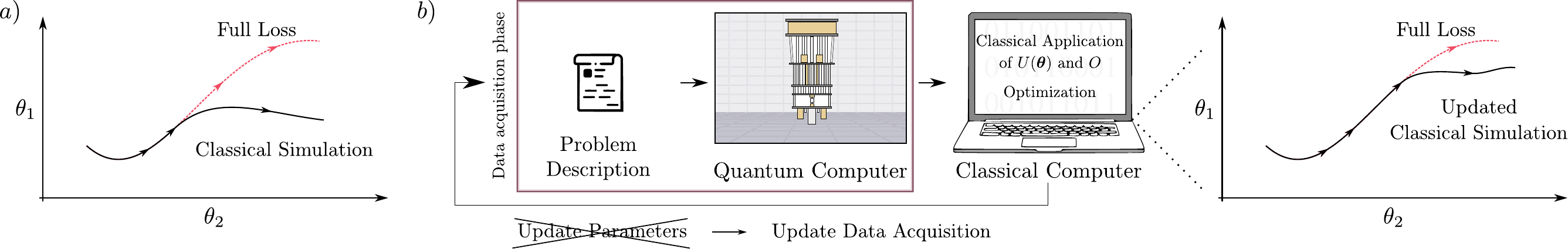}
    \caption{\textbf{Absence of barren plateaus, simulability, and faithful training.} a) We consider a problem in $\NBP$ such that, with high probability for some $\thv\sim\PC$, the circuit's adjoint action leads to a polynomial effective subspace as in Fig.~\ref{fig:modules}(b). Training on the classically estimated loss can be faithful for the first few optimization steps. However, as the optimization continues, a vital contribution to the loss could arise from operators not in $\BC_\lambda$. If this occurs, training on the classically estimated loss can be unfaithful and one can converge towards a parameter region that does not correspond to the minima of the true `full' loss function. b) To mitigate the aforementioned issues, one could use an alternative form of hybrid variational quantum computing where multiple, iterative, data acquisition steps are used. This information is then used by a classical computer to update a classical simulation of the loss and make it more faithful as the optimization progresses. As schematically shown, this scheme could make training the loss on a classical computer more faithful.        }
    \label{fig:modelsVQA2.0}
\end{figure*}

Second, we highlight the fact that being able to classically simulate and train a model without the need for expensive computational resources could be extremely useful. In particular, it is well known that there exist problems where one can classically train a variational state by minimizing some expectation value, but sampling from such a state on a classical computer is prohibitively expensive (such as in certain optimization problems or generative modeling tasks)~\cite{medvidovic2021classical,diez2023quantum,hadfield2018quantum,sreedhar2022quantum,krovi2022average, farhi2016quantum} or may not be supported by the particular simulation algorithm. Similarly, one can envision tasks where the goal is to variationally prepare some state for a quantum sensing task~\cite{endo2020variational,beckey2020variational,huerta2022inference}. Here, the ultimate goal is to obtain a metrological advantage rather than a computational one, meaning that classical simulations could be beneficial to save precious quantum resources. This paves the way for finding tasks where one ultimately cares about being able to prepare or sample from a state, as here one can train the parameters on a classical computer, transfer them to a quantum device (e.g., via Refs.~\cite{ran2020encoding,rudolph2022decomposition}), and perform additional operations and measurements therein.  

Third, and as discussed in the Supplemental Information, there appears to be no `one simulation algorithm' to rule them all, as we have found that different loss functions require different simulation algorithms. Hence, given some barren plateau-free model, we envision that asking the question `\textit{How can we simulate this loss?}' 
will lead us to new data-driven and quantum-inspired simulation algorithms. In the same line of thought, we think that our quest to better understand what makes an initial state well-aligned with the relevant subspaces, and thus a barren plateau-free loss, can allow us to uncover what are the true resources that make a loss function more difficult to simulate, but also more concentrated. For instance, the recent work of~\cite{tindall2024confinement} studied why certain quantum experiments~\cite{kim2023evidence} could be efficiently simulated via tensor networks~\cite{tindall2023efficient}. Therein, the authors showed that the dynamics generated by the quantum circuit ultimately live in a small (identifiable) subspace. In parallel, as most classical algorithms will have their limits where they particularly excel or perform poorly, it may also turn out that quantum methods where one runs the circuit on the device could still provide the most flexible `catch all' simulation method.  

As an example, we refer the reader to Ref.~\cite{diaz2023showcasing}, where a connection between computational resources in fermionic linear optic circuits, absence of barren plateaus, and simulability is unraveled. Therein, it was shown that when the quantum resources (measured by fermionic entanglement) are increased, the loss function becomes more concentrated. 
Simultaneously, even if the state has a large component in the polynomial subspace one still needs to estimate it, and whether such computation can be classically performed will likely determine if the loss is in $\CSIM$ or $\QESIM$ (see Claim~\ref{claim2}).  We thus leave for future work a detailed study of the initial states and measurement operators for which the problems are actually in $\CSIM$. Note that this analysis could promote simulations that are available on average, to faithful simulations for all $\thv$. For instance, a quantum convolutional neural network with product initial states can be efficiently simulated via tensor networks~\cite{cong2019quantum}, and hence is actually in $\CSIM$ $\forall \thv$. We believe that such pursuits will lead to deep insights into architecture-dependent resource theories and the limits of classical simulation algorithms.

Going further, we have noted that problems which are in $\QESIM$ but not in $\CSIM$ could be precisely those that are more amenable for implementation in the near term. This is due to the fact that the data acquisition might require shallower circuits than the implementation of $U(\thv)$, and thus be less impacted by hardware noise. Moreover, optimizing directly on the classical device will generally be much more efficient given the relative speed of classical devices and the possibility of using tools such as automatic differentiation through gradient back-propagation. 
Indeed, while the loss is defined as a unitary evolution of the state's projection onto the polynomial subspace, once the data is stored classically, one can potentially analyze it with more general procedures such as sending it through an appropriate classical neural network. This then begs the question of whether the best way to analyze the tomographic data is via a unitary map, or some more general function.

More fundamentally, the very fact that a study of barren plateaus has led us to  tomographic techniques~\cite{huang2020predicting,elben2022randomized,anshu2023survey} is in itself extremely interesting. Here, one can wonder about the limits of measure-first algorithms where initial data acquisition phases are allowed. While it is known that generically using the same shadows protocol for every problem has its limitations~\cite{gyurik2023limitations}, our results indicate that for every barren plateau-free problem there is an associated tomographic procedure, and that these protocols can vary widely from one task to the other. The extent to which these data acquisition methods are connected, and what their limitations are, is therefore an open question.

Regarding the utility of variational quantum computing and machine learning for classical data~\cite{wiebe2014quantumdeep,schuld2015introduction,biamonte2017quantum,cerezo2022challenges, di2023quantum}, our results should help us better understand if the quantum computer is being used as a meaningful form of enhanced feature-space. Namely,  if the problem's loss function is barren plateau-free, and we can identify how to simulate it, one should ask whether the loss is in $\CSIM$ or in $\QESIM$. Answering these questions could rule out situations where the use of a quantum computer is not necessary. 
One could further explore the parallels here with recent research into the possibility of classical surrogates for quantum models using classical data~\cite{schuld2021quantummodels,schreiber2022classical,landman2022classically,jerbi2023power,sweke2023potential}.

As discussed in the previous section, we know that if the problem has effective subspaces, then one could potentially simulate and train the loss given some initial set of parameters, but that the loss might not be faithful for the full optimization procedure (see Fig.~\ref{fig:modelsVQA2.0}(a)). To mitigate this issue, we envision a novel form of hybrid quantum computation as shown in  Fig.~\ref{fig:modelsVQA2.0}(b). Here, there is still a feedback loop between the classical and quantum device (similar to that of standard variational methods), but  instead of updating parameters in a quantum circuit, one updates the data acquired from the initial state. As such, as the optimization progresses, one uses the quantum computer to take new measurements on the initial state,  and then uses this information for a more faithful, updated classical simulation.  While we have not tested the performance of this protocol, we believe it might be useful and encourage the community to try it out.

To finish, we note that many schemes are not covered by our results. These include certain clever initialization strategies and many quantum learning models (such as some schemes in generative modeling). Moreover, we have shown that by drawing inspiration from conventional quantum algorithms one can construct unorthodox variational settings where  manipulation of exponentially large objects is enabled,  barren plateaus can be avoided,  and a quantum advantage is still possible.
In this regard, we note that under widely held cryptographic assumptions it is believed that BQP $\nsubseteq$ BPP/poly. Ref.~\cite{gyurik2023exponential} uses this fact in a PAC-learning setting to argue that many-body physics problems that are BQP-complete could be used to construct learning problems that are not classically learnable. There are many differences between their setting and ours and so whether or not similar arguments could be applied here requires careful work. Nonetheless, this cautiously hints towards the exciting possibility that there could still be useful situations where the implementation of parametrized circuits on quantum hardware can be used to achieve an exponential advantage. 

\section{Conclusions}

The arrival of variational quantum algorithms effectively democratized the world of quantum computing. Whereas coming up with new conventional quantum algorithms requires a careful consideration of how best to manipulate and extract quantum information, proposing a new variational quantum algorithm is relatively straightforward. One can simply identify a potential loss function and a circuit ansatz to use. The hard work of minimizing that cost is offloaded to a classical optimizer, reducing the burden both on the quantum computer and on the quantum researcher, or so it was hoped. 

While the hopes for this approach were largely fueled by the overwhelming success of neural networks in classical machine learning, it did not take long for the community  to realize that variational quantum computing can be significantly more challenging than its classical counterpart. In fact, as any practitioner can testify, actually performing the  optimization for moderate-sized problems is frustrating. This opened the door to studying causes of untrainability, leading to the discovery of several sources of barren plateaus and how to avoid them. While this quest was in many ways fruitful, with numerous approaches identified, here we have argued that all of these approaches are at heart rather simple. More concretely, we have argued that the strategies for avoiding barren plateaus considered by the community so far lead to algorithms that effectively live in polynomially sized subspaces. From here, one can  do away with parameterized quantum circuits and instead simulate the algorithms classically (potentially after an initial data acquisition stage on quantum hardware).

Our argument here at times borders on trivial, and may be---in some form, or another---already known to many. However, the connection between the absence of barren plateaus and simulability still has not permeated the field's Zeitgeist, and it is not uncommon to find claims where the absence of barren plateaus is equated with practical usefulness.   
In this manner, our conclusions push against current practice in the community and the net result is that variational quantum algorithms need a rethink. 

\textit{So, where does this leave us?} One path forward is to embrace our case-by-case argument as something positive. While some classical simulations have been performed~\cite{bermejo2024quantum,angrisani2024classically,lerch2024efficient,anschuetz2024arbitrary,mele2024noise,shin2024dequantising,ermakov2024unified,miller2025simulation}, much work remains to be done in this field of ``quantum machine learning dequantization''. Moreover, the fact that a quantum computer is generally not required to implement the parameterized quantum circuit, but might be needed for an initial data acquisition stage, makes such algorithms much easier to implement. To take advantage of this opportunity will require research into how to best perform the data acquisition stage and how to best perform the classical simulation needed to optimize the loss. Moreover, careful analytic work will be required to understand the scaling of the classical simulation algorithms. It may well turn out that simulating the parametrized quantum circuit on quantum hardware enjoys polynomial advantages over our best classical algorithms. 

Alternatively, one could rally against our arguments here and strive harder to identify potential avenues for achieving a provable exponential quantum advantage with parameterized quantum circuits. Of course, there are numerous classically hard quantum circuits and so it's natural to ask: \textit{Can't this be used to show that there are variational quantum algorithms that are barren plateau-free but not classically simulable even with data from quantum experiments?} Indeed, we showed that this is possible for a few contrived examples. The key question is whether the non-simulability of quantum circuits can be used more generally to construct useful, trainable, and non-classically simulable variational quantum algorithms. To achieve this we should seek to draw inspiration from conventional fault-tolerant quantum algorithms, and carefully consider how best to manipulate information in the exponentially large spaces in which quantum operators live. Variational quantum computing may still be a viable line of research in this regard but it will require a principled approach and a healthy dose of imagination. 

Finally, our focus here has been on analytically studying the scaling of variational quantum algorithms because we simply do not currently have the hardware to study how these algorithms perform for problem sizes we are actually interested in. However, there might be cases where, although we cannot prove that an algorithm does not have a barren plateau, training is possible. Indeed, this is the case for classical machine learning, whose heuristic success goes well beyond what can be guaranteed analytically. As such, there could be architectures which we do not know how to simulate, and where we cannot prove they are barren plateau-free, and yet seem to work in practice. 

We hope this perspective encourages the community to pause, reflect and take a more principled approach to variational quantum computing.

\section{Acknowledgments}

We are extremely grateful to Hsin-Yuan Huang for his invaluable contributions to this work. 
We thank Andrew Sornborger, Lukasz Cincio, Nathan Wiebe,  Chae-Yeun Park, Nathan Killoran, Maria Schuld,    Xanadu's Toronto office staff, and the QTML 2023 community for thoughtful and insightful conversations.  
M.C. acknowledges support from Los Alamos National Laboratory (LANL) ASC Beyond Moore’s Law project.
M.L. was supported by the Center for Nonlinear Studies at LANL. 
M.C., D.G.M., N.L.D. and P.B.  were supported by Laboratory Directed Research and Development (LDRD) program of LANL under project numbers 20230527ECR and 20230049DR.
Also, N.L.D. acknowledges support from CONICET  Argentina,
and P.B acknowledges support of DIPC.
A.I. acknowledges support by the U.S. Department of Energy (DOE) through a quantum computing program sponsored by the LANL Information Science \& Technology Institute and by the U.S. DOE, Office of Science, Office of Advanced Scientific Computing Research, under Computational Partnerships program.
E.F. acknowledges the support of the UK department for Business, Energy and Industrial Strategy through the National Quantum Technologies Programme, and the support of an industrial CASE (iCASE) studentship, funded by the Engineering and Physical Sciences Research Council (grant EP/T517665/1), in collaboration with the University of Strathclyde, the National Physical Laboratory, and Quantinuum. 
E.R.A.\ acknowledges support from the Walter Burke Institute for Theoretical Physics at Caltech.
S.T. and Z.H. acknowledge support from the Sandoz Family Foundation-Monique de Meuron
program for Academic Promotion. ST further acknowledges the grants for development of new faculty staff, Ratchadaphiseksomphot Fund, Chulalongkorn University [grant number 3230120336 DNS\_68\_052\_2300\_012], as well as funding from National Research Council of Thailand
(NRCT) [grant number N42A680126]. 
M.C., Z.H., and E.R.A. thank the organizers of the PennyLane Research Retreat, where part of this work was undertaken, for their hospitality.
This material is based upon work supported by the U.S. Department of Energy, Office of Science, National Quantum Information Science Research Centers, Quantum Science Center (LC).
This work was also supported by the Quantum Science Center (QSC), a National Quantum Information Science Research Center of the U.S. DOE.

    \section*{AUTHOR CONTRIBUTIONS}
The project was conceived by MC, ML, and ZH. Theoretical results were derived by MC, NLD, and ERA. All authors contributed 
to the manuscript review process. 

\section*{DATA AVAILABILITY}
No data was generated.

\section*{COMPETING INTERESTS}
The authors declare no competing interests.

\clearpage

\clearpage
\newpage

\makeatother
\onecolumngrid
\renewcommand\appendixname{Supp. Info.}
\renewcommand{\figurename}{Supplementary Figure}
\appendix

\section*{Supplemental Information for ``Does provable absence of barren plateaus imply classical simulability?}

\section{Barren plateau-free models, polynomial subspaces, and simulation algorithms}

We now present an analysis for each architecture and technique in Tables 1 and 2. In all cases we begin by defining the problem class of interest, reviewing the proof of absence of barren plateaus (i.e., determining the measurement operators and initial states leading to non-exponential concentration), and using this information to find the polynomially-large subspace. We then present tomographic techniques that can allow us to obtain the projection of the initial state onto the relevant subspace, and some appropriate simulation algorithms that can be used to classically compute the loss. In particular, here we show that that loss functions in  Tables 1 and 2 are in $\QESIM$.

\subsection{Small dynamical Lie algebras}\label{apx:smalldla}

Since several examples in Tables 1 and 2 fall inside of this category, we find it convenient to first present the abstract case of circuits with small dynamical Lie algebras.

We begin by (very briefly) reviewing some preliminary concepts that will be useful throughout the rest of the Appendices. We note that the goal of this perspective article is not to delve on mathematical details. As such, we will simply define the \textit{dynamical Lie algebra} of a circuit, and also what an \textit{equivariant} operator is. We refer the reader to the cited articles for further details. 

First, let us assume that the parametrized quantum circuit takes the form 
\begin{equation}\label{ap-eq:basicansatz}
    U(\thv)=\prod_{l=1}^Le^{-i \theta_l H_l}\,,
\end{equation}
where $\thv=(\theta_1,\theta_2,\ldots)$ are the trainable parameters, and $H_l$ Hermitian operators taken from some set of generators $\GC$. We note that Eq.~\eqref{ap-eq:basicansatz} can also capture an ansatz with fixed, un-parametrized, layers by setting some of the parameters to be constant. It is known that one can measure the potential expressivity of a circuit by analyzing its dynamical Lie algebra~\cite{zeier2011symmetry,larocca2021diagnosing,wiersema2023classification,kokcu2024classification,aguilar2024full,kazi2024analyzing}
\begin{equation} \label{ap-eq:dla}
    \mathfrak{g}=\langle i\GC\rangle_{{\rm Lie}}\,,
\end{equation}
where $\langle \cdot \rangle_{{\rm Lie}}$ denotes the Lie closure. As such, $\mathfrak{g}\subseteq\mathfrak{u}(2^n)$ is the real vector space spanned by the nested commutators of the elements in $i\GC$. Here $\mathfrak{u}(2^n)$ denotes the unitary algebra on the Hilbert space of $n$ qubits. The importance of $\mathfrak{g}$ stems from the fact that any unitary produced by the circuit (for any number of layers $L$, and for any set of parameters $\thv$) is contained in the subgroup generated by $\mathfrak{g}$, i.e., $U(\thv)\in e^{\mathfrak{g}}$. 

Recently, it was realized that too expressive circuits (where $\mathfrak{g}$ is exponentially large) are prone to  exhibit barren plateaus~\cite{holmes2021connecting,larocca2021diagnosing,fontana2023theadjoint,ragone2023unified}. As such, a significant amount of effort has been dedicated to understanding how the expressivity of the circuit could be reduced. One of the most promising approaches is to create circuits that respect a given symmetry group $S$~\cite{larocca2022group,meyer2022exploiting,skolik2022equivariant,nguyen2022atheory,ragone2022representation,zheng2021speeding,east2023all}. This led to the concept of \textit{equivariant}  circuits, where all the circuit's generators commute with the all the elements in $S$. More generally, we say that any operator is equivariant if it  commutes with  all the elements in $S$.

With the previous, definitions in mind, we denote by $\CC_{\poly\!\mathfrak{g}}$ the class of all instances where the circuit's dynamical Lie algebra  $\mathfrak{g}$ is simple and of polynomial dimension, i.e., $\dim(\mathfrak{g})\in\OC(\poly(n))$. We take  $\rho$ to be an $n$-qubit state preparable by a circuit with $\OC(\poly(n))$ gates when acting on the all-zero state, and $O$ some Hermitian operator. All the gates in the circuit are parametrized, and every parameter is sampled uniformly at random. We further assume that either $\rho$ or $O$ are in $i\mathfrak{g}$.

Recently,  an exact formula for deep versions of the circuits in  $\CC_{\poly\!\mathfrak{g}}$ was simultaneously obtained in Refs.~\cite{fontana2023theadjoint,ragone2023unified}. In particular, these works show that the variance is inversely proportional to $\dim(\mathfrak{g})$, as conjectured in Ref.~\cite{larocca2021diagnosing}. Moreover, if the projection of $\rho$ and $O$ in $i\mathfrak{g}$ is at most polynomially vanishing with the number of qubits, then it is guaranteed that a barren plateau will not appear. That is,
\begin{equation}
    \CC_{\poly\!\mathfrak{g}} \subset \NBP\,, \text{if the projection of $\rho$ or $O$ in $i\mathfrak{g}$ is in $\Omega(1/\poly(n))$}\,.
\end{equation}

Notably, the relevant subspace in this case is the DLA itself. This follows from the fact a Lie algebra is closed under the adjoint action of its associated Lie group, i.e., if $g\in\mathfrak{g}$, then  $U g U^\dagger \in \mathfrak{g}$ for all  $U \in e^{\mathfrak{g}}$. Consequently, if $O$ (or $\rho$) is in $i\liea$, the adjoint action of the group on it will result in an operator that lives inside the polynomially sized subspace. Hence,
\begin{equation}
    \CC_{\poly\!\liea}\subset \CC_{{\rm polySub}}\,.
\end{equation}

Next, one needs to obtain the projection of $\rho$ and $O$ into the dynamical Lie algebra (denoted as $\rho_{\mathfrak{g}}$ and $O_{\mathfrak{g}}$, respectively), which requires computing the inner product of the initial state and the measurement operator with the elements of  a basis of $i\mathfrak{g}$. Such computation will generally be case dependent, but or all cases considered we have always found an efficient tomographic procedure (see below for more information). From here, we can employ a recently introduced classical simulation algorithm called $\mathfrak{g}$-sim~\cite{somma2006efficient,goh2023lie}, which allows us to  (back-)propagate the elements of the dynamical Lie algebra with unitaries of the Lie group.  For this algorithm to work, one needs to compute the structure constants of the algebra. Importantly,  computing these constants takes a time that scales polynomially in $\dim(\liea)$. Thus, this suffices to show that 
\begin{equation}
    \CC_{\poly\!\liea} \subset  \QESIM\,.
\end{equation}
We remark here that even if obtaining the structure constants takes polynomial time, such  computational cost may still be prohibitively large~\cite{anschuetz2022efficient}. It is then possible that polynomial speed-ups could be obtained by running the quantum circuit in the quantum computer instead.

Here we find it important to note that we have assumed above that the Lie algebra is simple. However,  $\mathfrak{g}$ will generally decompose as a direct sum of simple Lie algebras (plus an abelian component) as $\mathfrak{g}=\oplus_\eta \mathfrak{g}_\eta$. Here, barren plateaus can be avoided (and thus simulations are enabled) if at least one of the simple components is polynomially large, and $\rho$ and $O$ have large projections therein. 

As we will see below, the following examples from Tables 1 and 2 actually avoid barren plateaus by having a small Lie algebra: $U(1)$-equivariant; and  $S_n$-equivariant.

\subsubsection{$U(1)$-equivariant circuit}

First, let us note that by $U(1)$-equivariant circuits, we  mean circuits that preserve the Hamming weight of a computational basis state. Hamming weight-preserving circuits have seen use for quantum chemistry, condensed matter, nearest centroid classification, portfolio optimization and hedging~\cite{arrazola2022universal,johri2021nearest,lopez2022symmetric,cherrat2023quantum}.

We define as  $\mathcal{C}_{{\rm U(1)equiv}}$ the class of all instances where the circuit is $U(1)$-equivariant, and is composed of Hamming weight preserving local gates. We take  $\rho$ to be an $n$-qubit pure state with fixed Hamming weight $k$ preparable by a circuit with $\OC(\poly(n))$ gates when acting on the all-zero state, and $O$ some $U(1)$-equivariant operator. All the gates in the circuit are parametrized, and every parameter is sampled uniformly at random.

Several proofs of absence of barren plateaus for $U(1)$-equivariant have been presented in the literature~\cite{larocca2021diagnosing,monbroussou2023trainability,cherrat2023quantum}. Crucially, they all somewhat rely on the fact that the restriction of the dynamical Lie algebra into the subspace with $k$ Hamming weight can be, at most, of dimension $\OC(\binom{n}{k}^2)$, meaning that a necessary condition for the loss function to not exhibit barren plateaus is that the initial state's Hamming weight is not too large. Note that bitstrings with Hamming weight $n-k$ can trivially be mapped to bitstrings of Hamming weight $k$ by flipping all bits, and, while our following statements equally apply to Hamming weight $n-k$, we will focus on the case of Hamming weight $k\leq n/2$. As such, one finds  
\begin{equation}
    \mathcal{C}_{{\rm U(1)equiv}}\subset \NBP\,,\,\, \text{if the Hamming weight $k$ of $\rho$ is such that $k\in\OC(1)$}.
\end{equation}

The previous description allows us to readily identify the polynomially small subspace. Namely, since the initial state has fixed Hamming weight $k$ and since the unitary and measurement preserve such Hamming weight, then all of the dynamics occur in the Hilbert's subspace of states with Hamming weight $k$. Moreover, we know that the dimension of such subspace is $\binom{n}{k}$, which is in $\OC(\poly(n))$ provided that $k\in\OC(1)$. Hence, $\forall \thv$ one has
\begin{equation}
    \mathcal{C}_{{\rm U(1)equiv}}\subset \CC_{{\rm polySub}}\,.
\end{equation}

From here, we can perform tomography of $\rho$ and $O$ in a basis of  Hamming weight $k$ states to obtain an efficiently storable description in the polynomial-sized subspace. Moreover, while we can use $\mathfrak{g}$-sim to simulate the loss functions in $\mathcal{C}_{{\rm U(1)equiv}}$, here we can also obtain the action of $U(\thv)$ in the subspace by studying how each local gate acts on a computational basis state, a procedure that is highly efficient (see Ref.~\cite{kerenidis2021classical} for an example with Givens rotations). These ingredients then suffice to show that 
\begin{equation}
    \mathcal{C}_{{\rm U(1)equiv}}\subset \QESIM\,.
\end{equation}

\subsubsection{$S_n$-equivariant}

Many problems in physics, machine learning, and quantum metrology have properties that remain invariant under any possible permutation of the $n$ qubits in the system. As such, studying these problems variationally will require the use of parametrized quantum circuits that are equivariant with respect to $S_n$, the symmetric group acting on $n$ qubits via qubit permutation~\cite{larocca2022group}.

We define as  $\mathcal{C}_{{\rm Sn\,equiv}}$ the class of all instances where the circuit is composed of $S_n$-equivariant layers of one and two-qubit gates. We take  $\rho$ to be an $n$-qubit state preparable by a circuit with $\OC(\poly(n))$ gates when acting on the all-zero state, and $O$ some $S_n$-equivariant operator. All the gates in the circuit are parametrized, and every parameter is sampled uniformly at random. 

In Ref.~\cite{kazi2024analyzing} the authors studied the dynamical Lie algebra associated to the unitaries used in $\mathcal{C}_{{\rm Sn\,equiv}}$ and determined that $\dim(\mathfrak{g})\in \OC(n^3)$. In particular, the absence of barren plateaus for this problem class was studied and proved in Ref.~\cite{schatzki2022theoretical}, where, as expected, the authors found that
\begin{equation}
    \mathcal{C}_{{\rm Sn\,equiv}} \subset \NBP\,, \text{if the projection of $\rho$ in $i\mathfrak{g}$ is in $\Omega(1/\poly(n))$}\,.
\end{equation}
Since the measurement operator is in the algebra, and since the unitary cannot take $O$ outside of the algebra, we find that the subspace is proper for all $\thv$ and thus
\begin{equation}
   \mathcal{C}_{{\rm Sn\,equiv}}\subset \CC_{{\rm polySub}}\,.
\end{equation}

As discussed above, the simulation can be performed via $\mathfrak{g}$-sim provided that we can compute the components of the initial state and measurement operator in the Lie algebra. Notably, this information can be obtained via the tomographic procedure of permutation invariant shadows introduced in Ref.~\cite{toth2010permutationally}, indicating that 
\begin{equation}
    \mathcal{C}_{{\rm Sn\,equiv}}\subset \QESIM\,.
\end{equation}

\subsection{Matchgate circuit and small modules}\label{apx:matchgate}

Matchgate circuits, also known as fermionic linear optics circuits, are circuits whose gates can be mapped to free fermion evolution \cite{jozsa2008matchgates}.
Their parameterized version ~\cite{wan2022matchgate,de2013power,oszmaniec2022fermion,cherrat2023quantum, diaz2023parallel} 
has received considerable attention in variational quantum computing as a test bed for benchmarking different ideas~\cite{larocca2021theory}. 
Crucially, the study of quantum resources and simulability of matchgate circuits ~\cite{bravyi2004lagrangian,dias2023classical,cudby2023gaussian,gigena2015entanglement} 
has been recently related to the theory of barren plateaus \cite{diaz2023showcasing}.

We denote by $\CC_{{\rm match}}$ the class of all instances where the circuit is a parametrized matchgate circuit composed of parametrized rotations about the $z$ axis and parametrized entangling $XX$ gates on neighboring qubits. We take  $\rho$ to be an $n$-qubit state preparable by a circuit with $\OC(\poly(n))$ gates when acting on the all-zero state, and $O$ some Pauli operator. All of the gates in the circuit are parametrized, and every parameter is sampled uniformly at random.

The presence and absence of barren plateaus in parametrized matchgate circuits was originally empirically studied in Refs.~\cite{wiersema2020exploring,larocca2021diagnosing}, with a recent full characterization presented in Ref.~\cite{diaz2023showcasing}. In particular, it is well known that the dimension of these circuits' dynamical Lie algebra scales as $\OC(n^2)$. This already means that, if $O$ or $\rho$ are in $i\mathfrak{g}$, we know that barren plateaus can be avoided, as the algebra itself is the relevant polynomially small subspace. However, the underlying algebraic properties of this model allow us to precisely characterize the loss function variance for \textit{any} initial state and measurement operator. Here, the key realization is that under the adjoint action of the circuit's unitary, operator space $\BC$ decomposes into group-modules as $\BC=\bigoplus_{\eta=1}^{2n}\BC_\eta$, where the operators in $\BC_\eta$ can be exactly characterized as the product of $\eta$ Majorana operators, and thus is of dimension $\binom{2n}{\eta}$. Hence, if $\eta\in\OC(1)$ for $\eta\leq n$, or if $2n-\eta\in\OC(1)$ for $\eta> n$, one obtains $\dim(\BC_\eta)\in\OC(\poly(n))$. In what follows, we will assume for simplicity $\eta\leq n$. Note that these subspaces are proper for all $\thv$. Assuming that $O$ is in some $\BC_\eta$ we get 
\begin{equation}
    \mathcal{C}_{{\rm match}}\subset \NBP\,,\,\, \text{if the projection of $\rho$ in $\BC_\eta\in\Omega(1/\poly(n))$ and $\eta\in\OC(1)$}.
\end{equation}

Here, identifying the small subspace is straightforward as the group modules are, by definition, closed under the adjoint action of the group. This allows us to say that if $\eta\in\OC(1)$, then 
\begin{equation}
   \mathcal{C}_{{\rm match}}\subset \CC_{{\rm polySub}}\,.
\end{equation}
Next, we need to obtain the component of $\rho$ into the module $\BC_\eta$ (note that since $O$ is a Pauli in $\BC_\eta$, we already have all the information we need from the measurement operator). Here, one simply needs to prepare $\rho$ and measure the expectation value of all the Pauli operators that can be expressed as a product of $\eta$ distinct Majoranas. Such a procedure is efficient if $\eta\in\OC(1)$. Then, we can obtain the adjoint action of the unitary in each module by computing the commutator graphs~\cite{diaz2023showcasing} (or generalized structure constants) and using $\mathfrak{g}$-sim. The classical cost here scales polynomially with the dimension of the module, and is therefore polynomial in the number of qubits. Hence, we have
\begin{equation}
    \mathcal{C}_{{\rm match}}\subset \QESIM\,.
\end{equation}

\subsection{Shallow hardware efficient ansatz}\label{apx:hea}

Since the shallow hardware efficient ansatz was presented in the main text, we only recap it briefly here for completeness.  \\

First, we recall that as shown in Ref.~\cite{cerezo2020cost}, and also at the end of this Supplemental Information, one can prove that
\begin{equation}
    \CC_{{\rm shallowHEA}}\subset \NBP\,,\,\, \text{if $\rho$ follows area law and $O$ is local}.
\end{equation}
Moreover, since the  measurement operator is local, the bounded light cone structure of a shallow hardware efficient ansatz means that the adjoint action of the circuit can only map the measurement operator to Paulis acting on at most $2L\in\OC(\log(n))$  qubits. There are $\OC(\poly(n))$ such operators and so the resulting subspace is a proper polynomial-sized subspace for any $\thv$. Hence
\begin{equation}
     \CC_{{\rm shallowHEA}}\subset\CC_{{\rm polySub}}\,.
\end{equation}

To perform an efficient classical simulation  of the loss we can obtain the adjoint action of the circuit in the subspace by  `dropping' from $U(\thv)$ all the gates that are outside of the measurement's backwards light cone. Then, the final step is to determine $\rho_\lambda$.  If $\rho$ has a known classical representation, then one can classically find its projection onto the subspace by computing (e.g., via pen and paper or using classical computing power) the expectation value for all operators $P_j\in\BC_\lambda$. However, if $\rho$ is some state obtained at the output of a give circuit, these expectation values cannot generally be obtained classically. However, with access to a quantum computer one can efficiently estimate $\rho_\lambda$ using local tomography of the reduced states. Notably instead of performing full tomography, one can also perform standard Pauli classical shadow tomography~\cite{huang2020predicting,elben2022randomized} That is, we prepare $\rho$, and implement randomly selected single-qubit Clifford gates on each qubit, and measure them in the computational basis. Since we care about estimating the expectation value of polynomially many local operators, this tomographic procedure is efficient.  As such, we find 
\begin{equation}
     \CC_{{\rm shallowHEA}}\subset \QESIM\,.
\end{equation}
It follows that no problem based on shallow hardware efficient ans\"{a}tze with a local measurement can be outside of $\QESIM$. We note that this result was originally reported in Ref.~\cite{basheer2022alternating}, with similar ideas explored concurrently in the context of quantum process learning in Ref.~\cite{jerbi2023power}.

\subsubsection{Quantum generative modeling}\label{apx:generative}

Quantum generative modeling is a sub-field of quantum machine learning where the goal is to learn an unknown probability distribution underlying a classical dataset, with the hope that the trained model can be utilized to generate new data samples that closely match the original data~\cite{benedetti2019generative,coyle2020born,alcazar2020classical,benedetti2019adversarial,perdomo2018opportunities,zoufal2021generative}. While several schemes exist for quantum generative modeling, we will here focus on quantum circuit generative models, which include quantum circuit Born machines (QCBMs) and quantum generative adversarial networks (QGANs). In these models, one prepares a parametrized quantum state that encodes the desired distribution over computational basis measurements. It has been shown that these models can avoid barren plateaus by using losses that contain significant local contributions, which, in conjunction with a shallow hardware efficient ansatz, guarantees non-exponential concentration~\cite{rudolph2023trainability,letcher2023tight}.

We define as $\CC_{{\rm gen}}$ the class of all instances where the circuit is a shallow one-dimensional hardware efficient ansatz composed of $L\in\OC(\log(n))$ layers of two-qubit gates acting on neighboring-qubits in a brick-layered structure. Moreover, we take  $\rho$ to be the all-zero state, and $O$ to be a weighted sum of projectors onto computational basis states. Here, we further assume that all the gates in the circuit are parametrized, and every parameter is sampled uniformly at random.

The loss functions in $\CC_{{\rm gen}}$ can be simulated using the fact that the whole circuit can be efficiently represented via tensor networks~\cite{markov2008simulating}. In particular, the all-zero state, when evolved by a shallow one-dimensional hardware efficient ansatz up to $\OC(\log(n))$ depth, admits an efficient matrix product state~\cite{verstraete2008matrix} description (a proper subspace) for all $\thv$. That is,
\begin{equation}
   \CC_{{\rm gen}}\subset \CC_{{\rm polySub}}\,.
\end{equation}
Moreover, since there is no need for tomographic measurements, and since sampling from the output distribution of an MPS is efficient~\cite{ferris2012perfect,stoudenmire2010minimally}, one has 
\begin{equation}
    \CC_{{\rm gen}}\subset \CSIM\,.
\end{equation}

\subsection{Generic shallow locally circuits}

Here, by generic shallow locally circuits we refer to circuits where local two-qubit gates act according to some arbitrary topology, and where, there is a path that connects the input state with the measurement operator that crosses, at most, $\OC(\log(n))$ such gates. This architecture encompasses hardware efficient ansatz in one-dimension, quantum convolutional neural networks, but also more general circuits where gates can act according to a generic topology (i.e, 2D and beyond). Importantly, the fact that the gate connectivity is arbitrary means that light-cone arguments such as ones previously presented for the shallow hardware efficient ansatz do not hold. 

Let us then define as $\mathcal{C}_{{\rm slocal}}$ the class of all instances where the parameterized quantum circuit is built  from local gates acting on pairs of qubits according to some arbitrary topology.  We then take  $\rho$ to be an $n$-qubit state preparable by a circuit with $\OC(\poly(n))$ gates when acting on the all-zero state, and the measurement operator $O$ to be a local Pauli acting on a qubit at the circuit's output. We assume that the parameters are sampled such that each gate forms an independent $2$-design over $SU(4)$.

The proof of absence of barren plateaus for these architectures can be derived with the results in~\cite{napp2022quantifying, zhang2023absence}, but for completeness, we also present a proof at the end of this Supplemental Information (Appendix~\ref{app:arbtopabsenceofbp}). As shown therein, if the initial state is pure and satisfies an area law of entanglement, we find
\begin{equation}
    \mathcal{C}_{{\rm slocal}}\subset \NBP\,,\,\, \text{if $\rho$ follows an area law of entanglement}.
\end{equation}

Moreover, the scrambling nature of the circuit effectively suppresses the contribution of high weight Paulis and the dynamics occur in a subspace composed of Paulis with $\OC(1)$-bodyness qubits, so that 
\begin{equation}
   \mathcal{C}_{{\rm slocal}}\subset \CC_{{\rm polySub}}\,.
\end{equation}
For a rigorous proof see Ref.~\cite{angrisani2024classically}.
Then, the component of the initial state therein can be obtained via classical shadows, and the simulation follows as a direct application of the algorithm in~\cite{angrisani2024classically}, showing that 
\begin{equation}
    \mathcal{C}_{{\rm slocal}}\subset \QESIM\,.
\end{equation}

\subsection{Small angle initialization}\label{apx:smallangle}

A recurring theme in this and the following sections will be how to classically simulate the loss function values for effective polynomial subspaces for problem classes in $\QESIM$ without relying on strong symmetries as in Secs.~\ref{apx:smalldla} \& \ref{apx:matchgate} or small entanglement light cones such as in Sec.~\ref{apx:hea}. 
Broadly speaking, one can either \textit{forward} propagate the initial state $\rho$ (or measurements of it) through the circuit, or \textit{backward} propagate the measurement operator $O$ in order to estimate the loss $\ell_{\bm\theta}(\rho, O)$. 

If no tomography of the initial state is needed, i.e., if an efficient MPS or MPS representation is known, the approach of forward propagation can be especially efficient in practice. For quantum circuits with local entangling gates and at most $\OC(\log(n))$ layers, the bond dimension of the tensor network will not grow to be exponentially large~\cite{markov2008simulating}, but even beyond log-depth circuits, high quality approximation with polynomial bond dimension may be possible. If a polynomial number of measurements of a non-trivial initial state are given, e.g., using standard Pauli classical shadows~\cite{huang2020predicting}, we notice that we can represent each measurement as an efficient matrix product operator (MPO)~\cite{verstraete2006matrix,evenbly2011tensor} with bond dimension $\chi=1$.
The loss function $\ell_{\bm\theta}(\rho, O)$ then becomes an average over all measurements from the quantum computer and their corresponding forward propagated MPOs.

On the flip side, one can also backward propagate the target operator through the quantum circuit, which may be particularly efficient in practice if $O$ is relatively simple in some basis (e.g., the Pauli basis) and one has to collect measurements from a  non-trivial initial state $\rho$. This can again be done using tensor networks in the form of MPOs or using methods, such as Pauli Propagation~\cite{fontana2023classical,rudolph2023classical}, which work using the Pauli transfer matrix formalism~\cite{nemkov2023fourier,beguvsic2023simulating,beguvsic2023fast}.

\medskip

The above techniques could be used for the classical simulation of both small angle initializations (discussed in this section) and QCNNs (discussed in the next section). Small angle initializations have been an increasingly popular strategy aimed at replacing the most commonly used random angle initialization techniques. Here we will be focusing on two of them: Hamiltonian variational ansatz with small angles and Gaussian initialization in deep variational quantum circuits.

\subsubsection{Hamiltonian variational ansatz}

The Hamiltonian variational ansatz was proposed in Ref.~\cite{wecker2015progress} as a problem-inspired circuit to find the ground state of a given target Hamiltonian $H$. Here, one splits  $H$ into a sum of terms $H=\sum_{i=1}^M \alpha_i h_i$, and builds the circuit as $U(\thv)=\prod_{l=1}^L\prod_{i=1}^Me^{-i\theta_{l,i}H_i}$, leading to a circuit that somewhat resembles a Trotterized evolution. 

Let us define as $\mathcal{C}_{{\rm HVA}}$ the class of all instances where the parameterized quantum circuit is built is a Hamiltonian variational ansatz for a given $H$, which can be expressed as $H=\sum_{i=1}^M \alpha_i h_i$. We assume that each $h_i$ is a $k$-local operator. We then take  $\rho$ to be an $n$-qubit state preparable by a circuit with $\OC(\poly(n))$ gates when acting on the all-zero state and the observable to be local, and the measurement operator $O$ to be a local Pauli. We assume that the parameters are sampled in a region around zero whose width scales as $\OC(1/(nML))$.

The proof of absence of barren plateaus is given in Ref.~\cite{park2023hamiltonian}, where such parametrized quantum circuit is shown to be equivalent to a time evolution driven by a single $k$-local Hamiltonian. In this new framework, the loss can be shown to avoid barren plateaus if the initial state is pure and satisfies an area law of entanglement. Hence, we have
\begin{equation}
    \mathcal{C}_{{\rm HVA}}\subset \NBP\,,\,\, \text{if $\rho$ follows an area law of entanglement}.
\end{equation}
Given that  $U(\thv)$ implements a short time evolution, the adjoint action of $U(\thv)$ on $O$ lives, with high probability for $\thv \sim \mathcal{P}$ , in an effective polynomially-sized subspace of operators acting on $\OC(1)$ qubits (see the proof of Theorem~4 in Ref~\cite{lerch2024efficient} for a detailed proof). Thus, one has
\begin{equation}
    \mathcal{C}_{{\rm HVA}}\subset  \CC_{{\rm polySub}}\
.
\end{equation}
Since the operators in the small subspace are all acting on at most $\OC(1)$ qubits, we can utilize standard Pauli classical shadows on the initial state~\cite{huang2020predicting}. Then, since the parametrized quantum circuit can effectively be expressed as $U(\thv)\simeq e^{-i K \tau}$, where $K$ is a $k$-local Hamiltonian ($k\in\OC(1)$)  and $\tau\in\OC(1/n)$, we know that the entanglement generated by the circuit is bounded, and we can efficiently compute the loss function via tensor networks or Pauli Propagation methods (with time and sample complexity guarantees provided by Theorems 2 and 4 in Ref.~\cite{lerch2024efficient}). Hence,
\begin{equation}
\mathcal{C}_{{\rm HVA}}\subset \QESIM\,.
\end{equation}

\subsubsection{Gaussian initialization in deep variational quantum circuits}

Next, let us define as $\mathcal{C}_{{\rm GI}}$ the class of all instances where the parametrized quantum circuit is built employing $L$ layers of single qubit parametrized rotations followed by an entangling layer of controlled $Z$ gates. We take  $\rho$ to be an $n$-qubit state preparable by a circuit with $\OC(\poly(n))$ gates when acting on the all-zero state and the observable to be local, and the measurement operator $O$ to be a local Pauli. We assume that the parameters are sampled according to a Gaussian distribution centered around zero and with variance in $\OC(1/L)$.

The proof of absence of barren plateaus for this initialization technique was presented in Ref.~\cite{zhang2022escaping}, where the authors find that exponential concentration is avoided if $\rho$ satisfies an area law of entanglement.  Therefore, one has

\begin{equation}
\mathcal{C}_{{\rm GI}}\subset \NBP\,,\,\, \text{if $\rho$ follows an area law of entanglement}\,.
\end{equation}
When the parameters are sampled from a Gaussian distribution whose width decreases with the circuit depth, then the adjoint action of the unitary on the measurements operator  lives, with high probability for $\thv\sim\PC$, in a low magic subspace.  Concretely, as proven in Ref.~\cite{lerch2024efficient}, it belongs to an effective subspace composed of operators acting on $\OC(\text{poly}(n))$ qubits that can be efficiently classically identified using the small-angle Pauli Propagation algorithm. It follows that 
\begin{equation}
    \mathcal{C}_{{\rm GI}}\subset  \CC_{{\rm polySub}}\
.
\end{equation}

This trade-off between the range of parameter values and circuit depth renders the quantum circuit effectively simulated with low magic simulation algorithms~\cite{lerch2024efficient, beguvsic2024fast}. Again we might need to perform Pauli measurements on the initial state $\rho$. Putting these results together leads to 
\begin{equation}
\mathcal{C}_{{\rm GI}}\subset \QESIM\,.
\end{equation}

\medskip

We note that Ref.~\cite{wang2023trainability} also proves a trainability guarantee for small angle initialization in an ansatz composed of alternating controlled Z layers and single qubit rotations. One can show that this is classically simulable by a similar argument to the one above.

\subsection{Quantum convolutional neural networks}

Quantum convolutional neural networks (QCNNs) are a type of circuit architecture introduced in Ref.~\cite{cong2019quantum} which has seen widespread use in supervised classification problems on both classical and quantum data~\cite{cong2019quantum,caro2021generalization,liu2023model,hur2021quantum,umeano2023can}. A QCNN is composed of parametrized gates acting on nearest neighbors which are interleaved with pooling layers where qubits are either measured or traced-out. Due to the way this architecture is constructed, its depth always scales logarithmically with the number of qubits, and the measurement operator is always local. 

\subsubsection{Standard QCNN with tracing-out pooling layers}

Let us define as  $\mathcal{C}_{{\rm QCNN}}$ the class of all instances where the circuit is a QCNN composed of layers of two-qubit gates acting on nearest neighbors which are interleaved with pooling layers where half of the qubits are traced-out. We take  $\rho$ to be an $n$-qubit state preparable by a circuit with $\OC(\poly(n))$ gates when acting on the all-zero state, and $O$ some Pauli operator acting on the last two qubits of the QCNN. All of the two-qubit gates in the circuit are parametrized, and every parameter is sampled uniformly at random. We assume that randomly sampling the parameters leads to all of the gates forming independent local $2$-designs. 

The proof of absence of barren plateaus for QCNNs can be found in Ref.~\cite{pesah2020absence}, but can also be derived with the tools presented at the end of this Supplemental Information (Appendix~\ref{app:arbtopabsenceofbp}). Therein, the authors showed that  if all of the two-qubit gates in the circuit form independent local $2$-designs, then  the loss function will not exhibit a barren plateau provided that the reduced states of $\rho$ onto any pair of qubits are not exponentially close to the identity (e.g., if the initial state follows an area law of entanglement). Hence, we can find that
\begin{equation}
    \mathcal{C}_{{\rm QCNN}}\subset \NBP\,,\,\, \text{if $\rho$ follows an area law of entanglement}.\nonumber
\end{equation}

With the previous, let us try to identify what are the relevant polynomial-sized subspaces under the QCNN's adjoint action. First, note that there always exists a set of parameters $\thv_{j}$ such that the QCNN is able to transform the local Pauli operator $O$ into \textit{any} other Pauli $P_j$. That is,
\begin{equation}
  \forall P_j\,, \text{then }\exists \thv_j\,\, \text{such that} \,\,  \frac{1}{d}\inprod{U\ad(\thv_j)OU(\thv_j),P_j}=1\,.\quad 
\end{equation}
While this equation indicates that the adjoint action allows us to reach exponentially many Pauli operators, this does not mean that the adjoint action will typically lead to exponentially large subspaces. Indeed, one can verify that, with high probability for $\thv \sim \mathcal{P}$, $U\ad(\thv)OU(\thv)$ will live in an effective polynomially large subspace whose basis is composed of Pauli of operators acting on at most $\OC(1)$ qubits~\cite{bermejo2024quantum,angrisani2024classically}. Thus, with high probability for some $\thv\sim \PC$, the QCNNs adjoint action over the measurement leads to an effective polynomially-sized subspace, and so
\begin{equation}
    \mathcal{C}_{{\rm QCNN}}\subset \CC_{{\rm polySub}}\,.
\end{equation}

Given that the small subspace contains operators acting on a small number of qubits, we can again perform standard Pauli classical shadows on the initial state~\cite{huang2020predicting} to obtain a classical description.  Then, since  the circuit is composed of local gates, and contains at most $\OC(\log(n))$ layers, we can simulate the loss via tensor networks or any of the techniques outlined at the beginning of Sec.~\ref{apx:smallangle}. High-fidelity approximations of $O(\thv)$ likely require polynomial MPO bond dimension, and algorithms such as Pauli Propagation could employ a simple truncation based on operator weight which has for example been used in Refs.~\cite{beguvsic2023fast,rudolph2023classical}.  Thus, we find 
\begin{equation}
    \mathcal{C}_{{\rm QCNN}}\subset \QESIM\,.
\end{equation}

Indeed, we note that the recent work of~\cite{bermejo2024quantum} has explicitly implemented the previous classical simulation method, effectively dequantizing the quantum convolutional neural network architecture.

\subsubsection{QCNN with post-selected pooling layers}

Next we consider a  class of QCNNs where the loss function depends on the outcomes of the measurements in the pooling layers; we denote this class of QCNNs as $\mathcal{C}_{{\rm pl-QCNN}}$. This setting is typically studied in the context of error-correction, where one wants to learn a circuit that optimizes recovery given some noise channel and logical state~\cite{cong2019quantum}. Here, the circuit is a QCNN composed of layers of two-qubit gates acting on nearest neighbors, interleaved with pooling layers where half of the qubits are measured and used to control single-qubit unitaries. While simulation is only available with high probability for the case of $\mathcal{C}_{{\rm QCNN}}$, we claim one can simulate those in $\mathcal{C}_{{\rm pl-QCNN}}$ for all parameter values; that is,
\begin{equation}\label{eq:pool_qcnn_subset}
    \mathcal{C}_{{\rm pl-QCNN}}\subset \QESIM\,.
\end{equation}

We assume there are $p$ parameterized gates, each a $2$-qubit unitary; we denote the unitary implemented by the QCNN as $U_{\bm{\theta}}$, where $\bm{\theta}\in\mathbb{R}^p$. In the post-selected setting the loss is some function of the expectation values of observables of the form:
\begin{equation}\label{eq:o_def_post_qcnn}
O=\sum_{\bm{m}\in\mathcal{M}\subseteq\left\{0,1\right\}^{n-k}}O_{m_1,\ldots,m_{n-k}}\otimes\ket{\bm{m}}\bra{\bm{m}}
\end{equation}
for some subset $\mathcal{M}\subseteq\left\{0,1\right\}^{n-k}$ and $k$-local $O_{m_1,\ldots,m_{n-k}}$ (assumed to have bounded Hilbert--Schmidt norm). A typical example of this setting is when a QCNN is used to learn the decoder of a quantum error correcting code where, given a circuit $V$ encoding a logical $k$-qubit state $\rho_L$, one maximizes the probability of correcting $t$ errors after an error channel $\mathcal{N}$ is applied~\cite{cong2019quantum}:
\begin{equation}
f\left(\bm{\theta}\right)=\sum_{\substack{\bm{m}\in\left\{0,1\right\}^n:\\\left\lVert\bm{m}\right\rVert_1\leq t}}\Tr\left(\rho_L\otimes\ket{\bm{m}}\bra{\bm{m}}U_{\bm{\theta}}\mathcal{N}\left(V\left(\rho_L\otimes\ket{0}\bra{0}^{n-k}\right)V^\dagger \right)U_{\bm{\theta}}^\dagger\right).
\end{equation}
Here, the input to the QCNN is the $n$-qubit state $\rho:=\mathcal{N}\left(V\left(\rho_L\otimes\ket{0}\bra{0}^{n-k}\right)V^\dagger\right)$, $O_{m_1,\ldots,m_{n-k}}=\rho_L$, and (for constant $t$) $\left\lvert\mathcal{M}\right\rvert=\OC\left(n^t\right)$.

While QCNNs are defined via real parameters---i.e., their parameters are of infinite precision---by a net argument~\cite{caro2021generalization, jerbi2023power} we can assume that each of the $p$ parameterized gates is drawn from a finite set $S_{p,\epsilon}$ of $2$-qubit gates up to incurring a total error $\epsilon$ in diamond distance. More concretely, Theorem~1 of the Supplemental Information of Ref.~\cite{caro2021generalization} states that such a net of cardinality:
\begin{equation}
S_{p,\epsilon}=\OC\left(\poly\left(\frac{p}{\epsilon}\right)\right)
\end{equation}
exists. We thus assume this setting without loss of generality.

We now show Eq.~\eqref{eq:pool_qcnn_subset} via the following theorem.
\begin{theorem}[Classical simulability of QCNNs]
Fix $\epsilon,\delta=\operatorname{\Theta}\left(1\right)$. Consider a QCNN with $p$ parameterized $2$-qubit gates drawn from $S_{p,\epsilon}$, and consider:
    \begin{equation}\label{eq:qcnn_loss}
        \Tr\left(O U\rho U^\dagger\right),
    \end{equation}
where $O$ is as in Eq.~\eqref{eq:o_def_post_qcnn}, $U$ is a unitary implementing the QCNN parameterized by $p$ elements of $S_{p,\epsilon}$, and $\rho$ is an $n$-qubit state. Using
    \begin{equation}
        N=\OC\left(\frac{\left\lvert\mathcal{M}\right\rvert^2}{\epsilon^2}\left(p\log\left(\frac{p}{\epsilon}\right)+\log\left(\left\lvert\mathcal{M}\right\rvert\right)+\log\left(\delta^{-1}\right)\right)\right)
    \end{equation}
    copies of $\rho$, with probability at least $1-\delta$, one can construct a classical shadow representation of $\rho$ that allows one to estimate Eq.~\eqref{eq:qcnn_loss} to additive error $\epsilon$ over all parameter settings. This classical shadow representation can be constructed with no prior knowledge of $U$ or $O$.
\end{theorem}
\begin{proof}
    Consider a classical shadows description of $\rho$ using the random Clifford basis protocol of Ref.~\cite{huang2020predicting}, built using
    \begin{equation}
        N=\frac{300\left\lvert\mathcal{M}\right\rvert^2\left\lVert O_{\text{max}}\right\rVert_{\text{F}}^2}{\epsilon^2}\left(p\ln\left(S_{p,\epsilon}\right)+\ln\left(\left\lvert\mathcal{M}\right\rvert\right)+\log\left(2\delta^{-1}\right)\right)
    \end{equation}
    copies of $\rho$. Here, $\left\lVert O_{\text{max}}\right\rVert_{\text{F}}$ is the maximal Hilbert--Schmidt norm of the $O_{m_1,\ldots,m_{n-k}}$, assumed to be $\OC\left(1\right)$. By Ref.~\cite{huang2020predicting}, this classical shadows description of $\rho$ is of size $\poly\left(N\right)$. Theorem S1 combined with Proposition S1 of \cite{huang2020predicting} guarantees that, with probability at least $1-\delta$, this classical shadows description suffices to estimate the expectation values of
    \begin{equation}
        M=S_{p,\epsilon}^p\left\lvert\mathcal{M}\right\rvert
    \end{equation}
    observables with Hilbert--Schmidt norm upper-bounded by $\left\lVert O_{\text{max}}\right\rVert_{\text{F}}$ each to additive error $\frac{\epsilon}{\left\lvert\mathcal{M}\right\rvert}$ in classical time $\poly\left(N\right)$. All that remains is demonstrating that there are $M$ observables being measured with Hilbert--Schmidt norm upper-bounded by $\left\lVert O_{\text{max}}\right\rVert_{\text{F}}$ by this QCNN.

    Now consider the set of $\left\lvert\mathcal{M}\right\rvert$observables:
    \begin{equation}
        \left\{O_{m_1,\ldots,m_{n-k}}\otimes\ket{\bm{m}}\bra{\bm{m}}\right\}_{\bm{m}\in\mathcal{M}}.
    \end{equation}
    The Hilbert--Schmidt norm of each $O_{m_1,\ldots,m_{n-k}}$ is upper-bounded by $\left\lVert O_{\text{max}}\right\rVert_{\text{F}}=\OC\left(1\right)$ by definition. By additionally counting all $S_{p,\epsilon}^p$ choices for $U$, there are thus $S_{p,\epsilon}^p\left\lvert\mathcal{M}\right\rvert$ operators of the form:
    \begin{equation}
        U^\dagger\left(O_{m_1,\ldots,m_{n-k}}\otimes\ket{\bm{m}}\bra{\bm{m}}\right)U
    \end{equation}
    with Hilbert--Schmidt norm upper-bounded by $\left\lVert O_{\text{max}}\right\rVert_{\text{F}}$. The considered classical shadows description of $\rho$ thus suffices to estimate the expectation value of each $U^\dagger\left(O_{m_1,\ldots,m_{n-k}}\otimes\ket{\bm{m}}\bra{\bm{m}}\right)U$ to additive error $\frac{\epsilon}{\left\lvert\mathcal{M}\right\rvert}$ with probability at least $1-\delta$, and thus also of each $U^\dagger OU$ to additive error $\epsilon$.
\end{proof}

\subsection{Non-unital noisy-circuits}

Let us define as  $\mathcal{C}_{{\rm non-u}}$ the class of all instances where the circuit is composed of two-qubit gates on a one-dimensional lattice, interleaved by local (single-qubit) noise, with a final layer of single-qubit gates. Furthermore, we assume that the noise is non-unital and contractive~\cite{mele2024noise}. We take $\rho$ to be a generic $n$-qubit state, and the observable $O$ to be local. We assume that all parameters are sampled such that each gate forms an independent $2$-design. 

As shown in~\cite{mele2024noise}, the effect of the non-unital noise is to effectively make the circuit shallow by ``resetting'' the qubits after each layer of parametrized gates. Indeed, as proved in Theorem 1 of ~\cite{mele2024noise}, one can only consider the last $\OC(\log(n))$ layers of the circuit, after which one can use light-cone arguments to prove that
\begin{equation}
    \mathcal{C}_{{\rm non-u}}\subset \NBP\,.
\end{equation}
Similarly, since the measurement is local, the same light-cone arguments can be used to show that the Heisenberg-evolved measurement operator will act on operators acting on at most $\OC(\log(n))$ qubits. That is, 
\begin{equation}
   \mathcal{C}_{{\rm non-u}}\subset \CC_{{\rm polySub}}\,.
\end{equation}
Finally, by either using the classical simulation algorithms of Ref.~\cite{mele2024noise} in the case of 1D circuits, or those in Refs.~\cite{angrisani2024classically, schuster2024polynomial} for arbitrary topologies,  one obtains
\begin{equation}
    \mathcal{C}_{{\rm slocal}}\subset \QESIM\,.
\end{equation}

\subsection{Parameterized Dynamic Circuit}

Parameterised dynamic circuits make crucial use of \textit{feedforward} operations. A feedforward operation consists of a measurement on a qubit followed by a conditional gate: If the measurement outcome is $0$ then $U_0 = I$ is applied, if instead it is $1$, then $U_1$ is applied where
\begin{equation}
U_1 =
\begin{pmatrix}
\cos\varphi e^{-i\phi} & -i\sin\varphi \\
-i\sin\varphi & \cos\varphi e^{i\phi}
\end{pmatrix} \, .
\end{equation}
Let us define $\mathcal{C}_{{\rm dyn}}$ as the class of all instances where the circuit implements a channel $\mathcal{E}(\thv) = \mathcal{U}_d(\thv_d) \circ \ldots \circ \mathcal{U}_1(\thv_1)$, where each channel $\mathcal{U}_j(\thv_j)$ is a completely positive, trace preserving map on $n$ qubits. The operations in each $\mathcal{U}_j$ can be written as a composition of single-qubit and two-qubit operations in an arbitrary circuit topology.
It is further assumed that: \begin{enumerate}
    \item Every component $\theta_i \in \thv$ parameterizes only a single operation in $\mathcal{E}$ (i.e., the parameters are independent and uncorrelated).
    \item $\mathcal{E}$ is locally scrambling, i.e., is invariant under single-qubit random gates from a unitary $2$-design after every gate.
    \item $\mathcal{E}$ has a constant worst-case feedforward distance. That is, we define $f$ be the shortest path from a qubit measurement to a feedforward operation through the backwards light cone of the measurement and take $f \in O(1)$.  
    \item $\mathcal{E}$ has an average entangling power of $\alpha$ and average swapping power $\beta \geq 0$~\cite{ware2023sharp}.
\end{enumerate}
Let $O$ be a $k$-local Hamiltonian with $k \in \log(n)$ and  let the initial state be arbitrary.  \\ 

Then, as shown in Ref.~\cite{deshpande2024dynamic}, similarly to the case of non-unital noise, the effect of the resets is to effectively make the circuit shallow by ``resetting'' the qubits after each layer of parametrized gates. However, in contrast to non-unital noise, these resets can be implemented in a controlled manner. 
Indeed, as proven in Theorem 1 of ~\cite{deshpande2024dynamic}, one finds that
\begin{equation}
    \mathcal{C}_{{\rm dyn}}\subset \NBP\,.
\end{equation}
Since these circuits are locally scrambling, the arguments of Ref.~\cite{angrisani2023learning} also apply here, thus providing bounds on the approximation error by considering only the reduced subspace. In particular, the mid circuit measurements at the core of dynamic circuits can be viewed as controlled two qubit rotations followed by qubit resets and thus can be absorbed into the setting considered in Ref.~\cite{angrisani2023learning}. Thus, we have 
\begin{equation}
    \mathcal{C}_{{\rm dyn}}\subset \CC_{{\rm polySub}}\,,
\end{equation}
and 
\begin{equation}
     \mathcal{C}_{{\rm dyn}} \subset \QESIM\,. 
\end{equation}

\section{Examples of non-concentrated  but also non-simulable loss functions}\label{apx:counterexample}

In this section we showcase an example of a non-concentrated loss function which is not classically simulable. This example showcases that $\NBP\neq \QESIM$. 

\medskip

Let us begin by assuming that the parameterized quantum circuit takes the form,
\begin{equation}\label{eq:circ-Bool}
    U(\thv) = U \left(\bigotimes_{i=1}^n e^{-i \theta_i X_i}\right),
\end{equation}
where $U$ is a fixed polynomial-size quantum circuit. More specifically, we take $U$ to be a quantum circuit that implements a Boolean function $B(x): \{0, 1\}^n \rightarrow \{-1, 1\}$ over the first output qubit that is classically hard to simulate for random input bitstrings, even with advice strings~\cite{gyurik2023exponential}. Many cryptographic problems related to discrete logarithms are widely believed to satisfy these criteria because they are random self-reproducible. We take any one of them.

Next, let us focus on  $U$ and the corresponding  Boolean function $B(x)$.
The circuit $U$ implements $B(x)$ in the following sense:
\begin{equation}\label{ap-eq:BooleanFuncCirc}
    \left\| \Tr[ U\ketbra{x}{x} U^\dagger Z_1] - B(x) \right\| < 0.01\,,
\end{equation}
for all $x \in \{0, 1\}^n$, where $Z_1$ is the Pauli-Z operator acting on the first qubit. As previously stated, we will focus on the fact that  $B(x)$ is classically hard to simulate for uniformly random input bitstrings $x \in \{0, 1\}^n$, even given a polynomial-size advice string. Equivalently, this can be stated as saying that this class of Boolean functions $\{B(x)\}$ is \emph{not} in $\mathrm{BPP}/\mathrm{poly}$. 
As data from quantum experiments are a restricted form of advice strings, we will use this observation to argue that a loss function that estimates $B(x)$ for $x$ uniformly sampled from $\{0, 1\}^n$ is not $\QESIM$.

\medskip

Specifically, let us consider the loss function
\begin{equation}\label{ap-eq:BooleanLoss}
    \ell_{\thv}(O) = \bra{0^n} \left(\bigotimes_{i=1}^n e^{i \theta_i X}\right) U^\dagger Z_1 U \left(\bigotimes_{i=1}^n e^{-i \theta_i X}\right) \ket{0^n} \, ,
\end{equation}
where $O=Z_1$ and we have explicitly taken the initial state $\rho$ to be the all-zero state.
Let us further suppose that each $\theta_i$ is uniformly sampled from the \textit{discrete} set $\{0, \pi / 2, \pi, 3 \pi / 2, 2 \pi\}$ such that the parametrized quantum circuit prepares a random computational basis state. That is, $\left(\bigotimes_{i=1}^n e^{-i \theta_i X}\right) \ket{0^n} = \ket{x}$ where $|x\rangle $ is a random bitstring. We call this distribution $\mathcal{P}_{\rm disc}$. We will then define $\CC_{{\rm Bool}}^{\rm disc.}$ as the class of all instances where the circuit is of the form in Eq.~\eqref{eq:circ-Bool} where $U$ implements a classically hard Boolean function (even with advice strings) and $\thv \sim \mathcal{P}_{\rm disc}$. 

Since this circuit by assumption approximates $B(x)$ for all $x$, as stated in Eq.~\eqref{ap-eq:BooleanFuncCirc}, and estimating $\{B(x)\}$ for uniformly sampled $x$ is not in $\mathrm{BPP}/\mathrm{poly}$, and data from quantum experiments are a form of advice string, it follows that $\CC_{{\rm Bool}}^{\rm disc.}$ is not in $\QESIM$. Moreover, since $U$ is a polynomial-size circuit, the problem is in $\QSIM$.  That is, we have 
\begin{equation}
    \CC_{{\rm Bool}}^{\rm disc.} \subset(\QSIM \cap \NQESIM) \,,
\end{equation}
where $\NQESIM$ denotes the classes not in $\QESIM$.

It remains to show that this problem does not exhibit a barren plateau. To show this we start by highlighting that $B(x)$ has a large variance over the input bitstrings. To see this note that if the variance is smaller than $1 / 100$, then for at least a $9/10$ fraction of the inputs, $B(x)$ must give the same output (the constants are not optimized here, this is just for illustrative purposes). As a result, the error of a trivial classical algorithm that always gives the same output, would only be wrong on at most $1/10$ of the inputs. This contradicts the assumption that $B(x)$ cannot, with high probability, be approximately computed classically. Since $\ell_{\thv}(O)$ approximates $B(x)$, as stated in Eq.~\eqref{ap-eq:BooleanFuncCirc}, it follows that $\ell_{\thv}(O)$ has a large variance over $\mathcal{P}_{\rm disc}$. That is, 
\begin{equation}
    \CC_{{\rm Bool}}^{\rm disc.} \subset\NBP\,.
\end{equation}
Thus, $\CC_{{\rm Bool}}^{\rm disc.} $ provides an example of a problem class that is in $\NBP$ but not in $\QESIM$, as claimed. 

\medskip

However, this example is rather contrived. For one, it considers an initialization over a discrete set of parameters. This is a long way from the spirit of more variational quantum algorithms. Moreover, this discrete initialization is doing much of the heavy lifting in this example. To see this one can also consider a variant of the above problem where the parameters are sampled uniformly in the range $[0, 2 \pi]$. We will call this distribution $\mathcal{P}_{\rm cont.}$ and the corresponding problem class $\CC_{{\rm Bool}}^{\rm cont.}$.

To understand both the absence of barren plateaus and classical simulability of $\CC_{{\rm Bool}}^{\rm cont.}$ it is helpful to write the backwards propagated observable in Eq.~\eqref{ap-eq:BooleanLoss}, $A \equiv U^\dagger Z_1 U$, in terms of its local and global components. Specifically, for a small number $k$, we can split $A$ as follows:
\begin{equation}
    A = \underbrace{\sum_{|P| \leq k} \alpha_P P}_{A^L} + \underbrace{\sum_{|P| > k} \alpha_P P}_{A^H}.
\end{equation}
We note that the expectation of $A^L$ in any quantum state can be measured efficiently on quantum hardware in an initial data acquisition phase. 

Next we highlight that for $\thv \sim \mathcal{P}_{\rm cont.}$, the contribution of $A^H$ to the loss decays exponentially with the bodyness $k$~\cite{huang2022learning}. (Intuitively, this can 
broadly be understood as following from the fact that the expectation of any Pauli operator $P$ of bodyness $k$ in a typical random product state decays exponentially in $k$.)
Hence, it follows that with high probability for $\thv \in \mathcal{P}_{\rm cont.}$ the problem lives in an effective polynomial subspace and we have 
\begin{equation}
    \CC_{{\rm Bool}}^{\rm cont.}\subset \mathcal{C}_{\rm polySub} \,.
\end{equation}
It then follows from Claim 2 in the main text that $\CC_{{\rm Bool}}^{\rm cont.}$ is classically simulable in the sense that with high probability for $\thv \sim \mathcal{P}_{\rm cont.}$ one can provide a classical estimate of the loss defined in Eq.~\eqref{ap-eq:BooleanLoss}. 

\medskip

\textit{So what do we conclude from these examples?} 
Firstly, this is a concrete example of the limitations of an average case notion of simulability. The loss we have constructed here is easy to classically simulate \textit{on average} over the continuous parameter range $\mathcal{P}_{\rm cont.}$ after an initial data acquisition phase. However, on a special measure-zero set of parameters, the problem becomes classically hard. 
Secondly, while in a typical random product state the expectation of any given global operator vanishes with high probability, something subtle can happen for the case of random computational basis states. Namely, while on average the contribution of a typical global Pauli is exponentially small, there are exponentially many such terms and with the right structure these can add up to give a non-zero total contribution. This is what happens for $\CC_{{\rm Bool}}^{\rm dist.}$---the contribution from $A^H$ does not vanish and gives rise to a non-vanishing variance. However, in parallel, the exponential number of global terms in  $A^H$ makes its contribution impossible to classically simulate. This suggests smart initialization procedures may be the key to constructing efficiently trainable variational quantum algorithms that are not also classically simulable.

\section{Absence of barren plateaus in shallow local circuits with arbitrary topology}\label{app:arbtopabsenceofbp}

In this section we will show that that shallow hardware efficient ansatz composed of two-qubit gates acting an arbitrary topology do not posses barren plateaus. While these results can be, in a way, derived from those in~\cite{napp2022quantifying, zhang2023absence}, we add the full proof here for completeness.  

We begin by present a few useful definitions. First, let  $\HC=(\mathbb{C}^2)^{\otimes 2}$ be the Hilbert space of two qubits. In what follows we will denote the subsystem of the first qubit with an index $A$, and that of the second qubit with an index $B$ so that $\HC=\HC_A\otimes\HC_B$. Then let $\BC(\HC^{\otimes t})$ be the set of  bounded operators acting on $t$ copies of the aforementioned Hilbert space. We are interested in computing expectation values of the form
\begin{equation}\label{eq:t-moments}
    \tau^{(t)}_{U}(X)=\int_{SU(4)} d\mu (U)\, U^{\otimes t} X (U\ad)^{\otimes t}\,,
\end{equation}
where $d\mu (U)$ denotes the uniquely defined, normalized, Haar measure over the fundamental representation of  $SU(4)$, and where $X\in\BC(\HC^{\otimes t})$.

We know that the twirls of Eq.~\eqref{eq:t-moments} can be evaluated via the Weingarten calculus~\cite{mele2023introduction}, as $\tau^{t}(X)_U$ projects into the commutant of the $t$-th fold tensor representation of $SU(4)$. Here, we can use the Schur-Weyl duality to find that a basis of the commutant is the subsystem-permuting representation of $S_t$, the symmetric group of order $t$. For instance, for $t=1$, a basis of $S_1$ is $\{\id\}$ where $\id$ denotes the $4\times 4 $ identity acting on $\HC$. Hence, one finds
\begin{equation}\label{eq:1-moment}
    \tau^{1}_{U}(X)=\int_{SU(4)} d\mu (U)\, U X U\ad=\frac{\Tr[X]}{4}\id\,.
\end{equation}
Going further to  $t=2$, it is standard to take a  basis of $S_2$ given by $\{\id\otimes \id,{\rm SWAP}\}$. Here, ${\rm SWAP}$ denotes the operator that permutes states acting on each of the copies of the Hilbert space and whose action is given by 
 \begin{equation}
    {\rm SWAP}(\ket{\psi_1}\otimes\ket{\psi_2})=\ket{\psi_2}\otimes\ket{\psi_1}\,,
\end{equation}
for $\ket{\psi_1},\ket{\psi_2}\in\HC$. This leads to the standard result 
\begin{equation}
    \tau^{2}_{U}(X)=\frac{1}{15}\left(\Tr[X]-\frac{\Tr[X{\rm SWAP}]}{4}\right)\id\otimes \id+\frac{1}{15}\left(\Tr[X{\rm SWAP}]-\frac{\Tr[X]}{4}\right){\rm SWAP}\,.\label{eq:2-moment-yuck}
\end{equation}

While the previous form can, and has~\cite{vznidarivc2022solvable,deneris2024exact,braccia2024computing,belkin2023approximate,mittal2023local}, been used to study properties of random quantum circuits composed of local gates, it has the key disadvantage that the basis $\{\id\otimes \id,{\rm SWAP}\}$ is not orthogonal. As such, we will determine a more convenient basis to work in. For this purpose, we first re-arrange the Hilbert spaces so that 
\begin{equation}
    \HC^{\otimes 2}=(\HC_{A,1}\otimes\HC_{B,1})\otimes(\HC_{A,2}\otimes\HC_{B,2})\rightarrow(\HC_{A,1}\otimes\HC_{A,2})\otimes(\HC_{B,1}\otimes\HC_{B,2})\,,
\end{equation}
so that 
\begin{equation}
  \id\otimes \id\rightarrow \id_A\otimes \id_B\,,\quad   {\rm SWAP}\rightarrow{\rm SWAP}_A\otimes {\rm SWAP}_B\,,
\end{equation}
where $SWAP_{A(B)}$ are the operators that swap the $A(B)$ qubits in the two-fold tensor product of $\HC$, and $\id_{A(B)}$ the identity operator on the two-copies of the subsystem's Hilbert space $\HC_{A(B)}$. Moreover, let us write 
\begin{align}
    {\rm SWAP}_{A}=\frac{1}{2}\Big(\id_{A,1}\otimes \id_{A,2}+X_{A,1}\otimes X_{A,2}+Y_{A,1}\otimes Y_{A,2}+Z_{A,1}\otimes Z_{A,2} \Big)=\frac{1}{2}\left(\id_{A}+{\rm SW}_A\right)\,,\nonumber
\end{align}
where we defined 
\begin{align}
\id_{A}=\id_{A,1}\otimes \id_{A,2}\,,\quad 
    {\rm SW}_A=X_{A,1}\otimes X_{A,2}+Y_{A,1}\otimes Y_{A,2}+Z_{A,1}\otimes Z_{A,2}\,.
\end{align}
Combining the previous results, we obtain
\begin{align}
    {\rm SWAP}=\frac{1}{4}\Big(\id_A\otimes\id_B+{\rm SW}_A\otimes\id_B+\id_A\otimes {\rm SW}_B+{\rm SW}_A\otimes {\rm SW}_B\Big)\,.
\end{align}
From here, we can find that an orthonormal basis for the commutant of the two-fold tensor rep of $SU(4)$ is
\begin{equation}
    \{\frac{\id_A\otimes\id_B}{4},\frac{{\rm SW}_A+ {\rm SW}_B+{\rm SW}_A\otimes {\rm SW}_B}{4\sqrt{15}}\}\nonumber\,,
\end{equation}
where we have dropped the trivial identity terms for ease of notation. Then, we can write the two-fold twirl as 
\begin{equation}
    \tau^{2}_{U}(X)=\frac{1}{16}\Tr[X]\id_A\otimes \id_B+\frac{1}{240}\Tr[X({\rm SW}_A+ {\rm SW}_B+{\rm SW}_A\otimes {\rm SW}_B)]({\rm SW}_A+ {\rm SW}_B+{\rm SW}_A\otimes {\rm SW}_B)\,.\label{eq:2-moment}
\end{equation}

Given Eq.~\eqref{eq:2-moment} we can now compute the following twirls: 
\begin{align}
\tau^{2}_{U}(P_A\otimes P_B )&=\frac{1}{15}({\rm SW}_A+ {\rm SW}_B+{\rm SW}_A\otimes {\rm SW}_B)\label{eq:2-twirl-Pauli}\,,\\
    \tau^{2}_{U}({\rm SW}_A )&=\frac{1}{5}({\rm SW}_A+ {\rm SW}_B+{\rm SW}_A\otimes {\rm SW}_B)\label{eq:2-twirl-A}\,,\\
    \tau^{2}_{U}({\rm SW}_B )&=\frac{1}{5}({\rm SW}_A+ {\rm SW}_B+{\rm SW}_A\otimes {\rm SW}_B)\label{eq:2-twirl-B}\,,\\
    \tau^{2}_{U}({\rm SW}_A\otimes{\rm SW}_B )&=\frac{3}{5}({\rm SW}_A+ {\rm SW}_B+{\rm SW}_A\otimes {\rm SW}_B)\label{eq:2-twirl-AB}\,,
\end{align}
where $P_{A(B)}$ is a single-qubit Pauli acting on the subsystems $A(B)$.

\begin{figure}[t]
    \centering
\includegraphics[width=.2\linewidth]{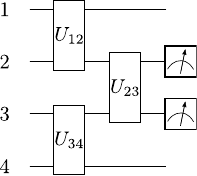}
    \caption{\textbf{Four qubit circuit example.} Consider the following circuit acting on four qubits. At the end of the circuit we measure a Pauli operator on qubits 2 and 3. }
    \label{fig:example}
\end{figure}

Using the previous equation, we can evaluate the loss function variance over arbitrary circuits composed of $SU(4)$ gates by simply keeping track of how the ${\rm SW}$ operators transform between themselves and how the coefficients accumulate. As an example, consider the four qubit circuit shown in Fig.~\ref{fig:example}, where the unitary $U$ implementing the whole circuit can be expressed as $U=U_{23}U_{34}U_{12}$. We assume the input to the circuit is $\rho$, and that at the end of the circuit we measure the expectation value of some Pauli operator acting on qubits 2 and 3. We further assume  that all the two-qubit gates form independent two-design over $SU(4)$. We want to compute a quantities  of the for $\mathbb{E}_U[\Tr[U\rho U\ad O]]$ and $\mathbb{E}_U[\Tr[U\rho U\ad O]^2]$. In particular we will evaluate the expectation value over the circuit as
\begin{align}
\mathbb{E}_U[f(U)]&=\int_{SU(4)}d\mu(U_{12})d\mu(U_{34})d\mu(U_{23})f(U)\,,
\end{align}
Using Eq.~\eqref{eq:1-moment}, we can readily show that  
\begin{equation}
    \mathbb{E}_U[\Tr[U\rho U\ad O]]=0\,.
\end{equation}
Then, since
\begin{align}
    \Tr[U\rho U\ad O]^2&=\Tr[U^{\otimes 2}\rho^{\otimes 2} (U\ad)^{\otimes 2} O]\nonumber\\
    &=\Tr[\rho^{\otimes 2} (U\ad)^{\otimes 2} OU^{\otimes 2}]\nonumber\,,
\end{align}
we find 
\begin{align}
    \mathbb{E}_U[\Tr[\rho^{\otimes 2} (U\ad)^{\otimes 2} OU^{\otimes 2}]]&=\Tr[\rho^{\otimes 2}\mathbb{E}_U[ (U\ad)^{\otimes 2} OU^{\otimes 2}]]\,,\nonumber
\end{align}
where we have used the linearity of the trace. Then, we see that we need to compute 
\begin{equation}
\tau_{U_{12}}^{(2)}\circ\tau_{U_{34}}^{(2)}\circ\tau_{U_{23}}^{(2)}[ O^{\otimes 2}]\,.
\end{equation}
Since  $O$ is some Pauli acting non trivially on qubits $2$ and $3$, we find via Eq.~\eqref{eq:2-twirl-Pauli}
\begin{align}
  \tau_{U_{23}}^{(2)}[ O^{\otimes 2}]  = \frac{1}{15}({\rm SW}_2 +{\rm SW}_3+{\rm SW}_2\otimes {\rm SW}_3)\,.\label{eq:test-1}
\end{align}
We then use Eq.~\eqref{eq:2-twirl-A}--\eqref{eq:2-twirl-AB} to obtain 
\begin{align}
    &\tau_{U_{34}}^{(2)}[\frac{1}{15}({\rm SW}_2 +{\rm SW}_3+{\rm SW}_2\otimes {\rm SW}_3)]\nonumber\\
    =&\frac{{\rm SW}_2}{15}+\frac{1}{75}({\rm SW}_3 +{\rm SW}_4+{\rm SW}_3\otimes {\rm SW}_4)+\frac{3}{75}({\rm SW}_2\otimes{\rm SW}_3 +{\rm SW}_2\otimes{\rm SW}_4+{\rm SW}_2\otimes{\rm SW}_3\otimes {\rm SW}_4)\label{eq:test-2}\,.
\end{align}
Here, we can already observe an important phenomenon. If a ${\rm SW}_{i}$ goes into a twirl  from a gate $U_{ij}$, then we obtain terms ${\rm SW}_{i}$, ${\rm SW}_{j}$, and ${\rm SW}_{i}\otimes {\rm SW}_{j}$ with a coefficient $\frac{1}{5}$.  Similarly, if ${\rm SW}_{i}\otimes {\rm SW}_{j}$ goes into a twirl  from a gate $U_{ij}$, then we obtain terms ${\rm SW}_{i}$, ${\rm SW}_{j}$, and ${\rm SW}_{i}\otimes {\rm SW}_{j}$ with a coefficient $\frac{3}{5}$. Hence, we simply need to keep track of indexes and coefficients.  As such, let us  we-write Eq.~\eqref{eq:test-1} as
\begin{equation}
    O^{\otimes 2}\xrightarrow[\tau_{U_{23}}^{(2)}]{} \left\{\left\{(2),\frac{1}{5}\right\},\left\{(3),\frac{1}{5}\right\},\left\{(2,3),\frac{1}{5}\right\}\right\}\,.
\end{equation}
Then, the next twirl over $U_{34}$ in Eq.~\eqref{eq:test-2} leads to 
\begin{align}
\xrightarrow[\tau_{U_{33}}^{(2)}]{} &\Big\{\left\{(2),\frac{1}{5}\right\},\left\{(3),\frac{1}{75}\right\},\left\{(4),\frac{1}{75}\right\},\left\{(3,4),\frac{1}{75}\right\},\nonumber\\
&\left\{(2,3),\frac{3}{75}\right\},\left\{(2,4),\frac{3}{75}\right\},\left\{(2,3,4),\frac{3}{75}\right\}\Big\}
\end{align}
and finally the next twirl over $U_{12}$
\begin{align}
\xrightarrow[\tau_{U_{33}}^{(2)}]{} &\Big\{\left\{(1),\frac{1}{75}\right\},\left\{(2),\frac{1}{75}\right\},\left\{(3),\frac{1}{75}\right\},\left\{(4),\frac{1}{75}\right\},\nonumber\\
&\left\{(1,2),\frac{1}{75}\right\},\left\{(3,4),\frac{1}{75}\right\},\left\{(2,3),\frac{3}{375}\right\},\nonumber\\
&\left\{(1,3),\frac{3}{375}\right\},\left\{(2,4),\frac{3}{375}\right\},\left\{(1,4),\frac{3}{375}\right\},\nonumber\\
&\left\{(1,2,3),\frac{3}{375}\right\},\left\{(1,2,4),\frac{3}{375}\right\},\left\{(2,3,4),\frac{3}{375}\right\},\nonumber\\
&\left\{(1,3,4),\frac{3}{375}\right\},\left\{(1,2,3,4),\frac{3}{375}\right\}\Big\}\,.
\end{align}
The previous example unveils a simple pattern that will be universally true. First, the ``most-local'' operators, are the ones which will be suppressed the least. Whereas ``more global'' operators will be  quickly suppressed. As such, we can intuitively see that, in average, the Heisenberg evolved operator will live on $\OC(1)$ bodied operators. Indeed, one can study how much the coefficient of a given operators at the out put of the circuit will look like by evaluating some weighted random walk over the circuit architecture following the rules mentioned above. 

Putting it all together, we are now ready to show that shallow local circuits avoid barren plateaus when measured locally at their output. Interestingly, this result encompassed, and constitutes a simpler proof, to those in Refs.~\cite{cerezo2020cost} and~\cite{pesah2020absence}. Consider therefore an arbitrary circuit with two-qubit gate topology
\begin{equation}
    \mathcal{A} = \{(i_k, j_k), \quad \forall k = 1, \ldots, K\}\,.
\end{equation}
Note that we are not assuming any dimensionality on $ \mathcal{A}$, meaning that it could be as general as desired. Then,  if the measurement operator is locally acting on the $\mu$-th qubit, and that this index appears in $\mathcal{A}$ $L$ times, one can lower bound the expectation value's variance as
\begin{align}
    \Var_{U}[\Tr[U\rho U\ad O]]&\geq \frac{1}{15}\frac{1}{5^{L-1}}\Tr[\rho^{\otimes 2}{\rm SW}_{\mu}]\nonumber\\
    &=\frac{1}{15}\frac{2}{5^{L-1}}(\Tr[\rho^{2}_\mu]-\frac{1}{2})\,,
\end{align}
where we have defined as $\rho^{2}_\mu$ the reduced state of $\rho$ on the $\mu$-th qubit, and where we have used the fact that ${\rm SW}_{\mu}=2({\rm SWAP}_\mu-\id)$. Moreover, here we have used the fact that we can drop all other terms in the path summation since the overlap between $\rho^{\otimes 2}$ and the tensor product of many ${\rm SW}$'s is always positive. More generally, if we define as $\SC$ as the set of indexes which are connected through a path of length $L$ from the beginning of the circuit to the measurement on the $\mu$-th qubit, one can tighten the bound as  
\begin{align}
    \Var_{U}[\Tr[U\rho U\ad O]]&\geq \frac{1}{15}\frac{2}{5^{L-1}}\sum_{j\in\SC}(\Tr[\rho^{2}_j]-\frac{1}{2})\,.
\end{align}
The previous result shows that if the circuit is shallow, i.e., if there exists a path of length $L$  such that $L\in\OC(\log(n))$; and if the input state $\rho$ is not too entangled, i.e.,  $|\Tr[\rho^{2}_j]-\frac{1}{2}|\in\Omega(1/\poly(n))$ then
\begin{align}
    \Var_{U}[\Tr[U\rho U\ad O]]\in\Omega\left(\frac{1}{\poly(n)}\right)\,.
\end{align}
Notably, we find that the reduced states will be exponentially close to the maximally mixed state if they satisfy a volume law of entanglement, whereas their distance to it will be polynomially small is they follow an area law of entanglement~\cite{leone2022practical}.

\end{document}